\documentclass{article}

\usepackage{arxiv}

\usepackage[utf8]{inputenc} 
\usepackage[T1]{fontenc}    
\usepackage{hyperref}       
\usepackage{url}            
\usepackage{booktabs}       
\usepackage{amsfonts}       
\usepackage{nicefrac}       
\usepackage{microtype}      
\usepackage{lipsum}
\usepackage{graphicx}
\usepackage{graphicx}
\usepackage{amsmath}
\usepackage{amsmath}
\usepackage{amssymb}
\usepackage{mathtools}
\numberwithin{equation}{section}
\usepackage{empheq}
\usepackage{physics}
\usepackage{color}
\usepackage[hang,small,bf]{caption}
\usepackage[subrefformat=parens]{subcaption}
\def\bC{\mathbb{C}}

\def\bR{\mathbb{R}}

\def\bZ{\mathbb{Z}}

\def\sn{\mathrm{sn}\,}
\def\ns{\mathrm{ns}\,}
\def\dn{\mathrm{dn}\,}

\def\x{\boldsymbol{x}}

\renewcommand{\Re}{\operatorname{Re}}
\renewcommand{\Im}{\operatorname{Im}}

\def\ii{\mathrm{i}}
\def\e{\mathrm{e}}
\def\sgn{\mathrm{sgn}\,}
\let\hat\widehat

\def\Xint#1{\mathchoice
   {\XXint\displaystyle\textstyle{#1}}%
   {\XXint\textstyle\scriptstyle{#1}}%
   {\XXint\scriptstyle\scriptscriptstyle{#1}}%
   {\XXint\scriptscriptstyle\scriptscriptstyle{#1}}%
   \!\int}
\def\XXint#1#2#3{{\setbox0=\hbox{$#1{#2#3}{\int}$}
     \vcenter{\hbox{$#2#3$}}\kern-.5\wd0}}

\def\dashint{\Xint-}

\def\Llleftarrow{%
\lower2pt\hbox{\begingroup
\tikz
\draw[shorten >=0pt,shorten <=0pt] (0,3pt) -- ++(-1em,0) (0,1pt) -- ++(-1em-1pt,0) (0,-1pt) -- ++(-1em-1pt,0) (0,-3pt) -- ++(-1em,0) (-1em+1pt,5pt) to[out=-105,in=45] (-1em-2pt,0) to[out=-45,in=105] (-1em+1pt,-5pt);
\endgroup}
}
\usepackage{amsthm}

\theoremstyle{definition}
\newtheorem{dfn}{Definition}[section]
\newtheorem{prop}[dfn]{Proposition}
\newtheorem{lem}[dfn]{Lemma}

\newtheorem{cor}[dfn]{Corollary}
\newtheorem{rem}[dfn]{Remark}

\title{Non-local time evolution equation with singular integral and its application to traffic flow model}

\author{
 Kohei Higashi \\
  \texttt{koheih@ms.u-tokyo.ac.jp} \\
}

\begin{document}
\maketitle
\begin{abstract}
We consider an integro-differential equation model for traffic flow which is an extension of the Burgers equation model. 
To discuss the model, we first examine general settings for integrable integro-differential equations and find that they are obtained through a simple residue formula from integrable eqations in a complex domain. 
As demonstration of the efficiency of this approach, we list several integrable equations including a difference equation with double singular integral and an equation with elliptic singular integral.
Then, we discuss the traffic model with singular integral and show that the model exhibits interaction between free flow region and congested region depending on the parameter of non-locality.  
\end{abstract}


\section{Introduction}
In dynamics of natural or social phenomena, non-local effect often plays an important role and has to be represented adequately for mathematical modeling. Integro-differential equations are an effective tool for this purpose and have been used in fields such as fluid mechanics, electric circuits, and epidemiology. 
A notable example in fluid dynamics is the intermediate long wave (ILW) equation which describes long internal gravity waves in a stratified fluid with finite depth \cite{Kubota}\cite{RIJoseph_1977}:
\begin{align}
\begin{cases}
    &\displaystyle\pdv{u(x,t)}{t}+\frac{1}{\delta}\pdv{u(x,t)}{x}+2u(x,t)\pdv{u(x,t)}{x} +   T\left(\pdv[2]{u(x,t)}{x}\right)= 0,\\[3pt]
    &\displaystyle T\left(v(x,t))\right):= \frac{-1}{2\delta}\dashint_{-\infty}^\infty \coth{\frac{\pi}{2\delta}(x-x')}v(x',t)\,dx'. \label{eq:ILW1}
\end{cases}
\end{align}
An important property of ILW eq. \eqref{eq:ILW1} is that it has $N$ soliton solutions, infinite number of conserved quantities, B\"{a}cklund transformation, and that initial value problems can be solved by inverse scattering transform (IST)\cite{RIJoseph_1977}\cite{Satsuma}\cite{Kodama1982}.  
Namely, it is a nonlinear integrable equation. 
In fact, by taking the limit $\delta \to 0$, \eqref{eq:ILW1} turns to the celebrated Korteweg deVries (KdV) equation and, by $\delta \to +\infty$, it turns to Benjamin-Ono equation\cite{Benjamin1967}\cite{Ono1975}.  

To obtain the solutions and conserved quantities for \eqref{eq:ILW1}, IST schemes were used and its associated spectral problems turned out to be the Riemann-Hilbert (RH) boundary value problem\cite{Kodama1982}.
Other RH problems have been introduced and integro-differential equations related to well-known integrable equations such as non-linear Schr\"{o}dinger equation\cite{ZakharovShabat1972}, Modified Korteweg-deVries equation\cite{Wadachi1972}, sine-Gordon equation\cite{Hirota1971} and Kadomtsev-Petiviashuvili equation\cite{KadomtsevPetviashvili1970} have been constructed\cite{DegasperisSantini1983}. 
The method is to obtain solutions of a RH boundary value problem, $\psi(z)$ ($z \in \mathbb{C}$) which satisfy a certain constraint such as $\psi(x)=-\psi (x+2\ii\delta)$ $(x \in \mathbb{R})$, and consider evolution equations which preserve the constraint.  
Most of these integrable integro-differential equations are expressed with singular integral as in \eqref{eq:ILW1}.

This approach with RH problems can be extended to an infinite series of integrable integro-differential equations, that is, integro-differential hierarchies of ILW, Sine-Gordon and AKNS equations \cite{DegasperisSantini1983}\cite{DegasperisSantiniAblowitz1985}\cite{SantiniAblowitzFokas1987}.
Another approach to construct the integro-differential analogue of ILW hierarchy and that of Intermediate nonlinear Schr\"{o}dinger equation hierarchy was proposed \cite{TutiyaSatsuma2003} on the basis of the theory of KP hierarchy\cite{Sato}. 
In Ref.~\cite{TutiyaSatsuma2003}, additional discrete flow in applied to the KP hierarchy
and the compatibility condition of the two flows are proved to give these integro-differential hierarchies.

Recently, Satsuma and Tomoeda proposed an integro-differential equation which describes time evolution of traffic density $\rho(x,t)$ as
\begin{align}\label{eq:Satsuma-Tomoeda}
\frac{\partial\rho}{\partial t}=V_{\max}\left(\dashint_{-\infty}^\infty \coth \frac{\pi}{2 \delta}(x-y)\frac{\partial \rho}{\partial y}(y)dy +I-\rho_{\max}\right)\frac{\partial \rho}{\partial x}+D\frac{\partial^2 \rho}{\partial x^2}.
\end{align}
Here $V_{\max}$ is the maximum velocity of a car, $\rho_{\max}$ is the density of the deadlock phenomenon, $\rho(\pm\infty,t)$ are assumed to be constant in time $t$ and $I:=\rho(\infty,t)+\rho(-\infty,t)$.
Equation \eqref{eq:Satsuma-Tomoeda} is an extension of the equation:
\begin{align}\label{eq:traffic_Burgers}
\frac{\partial\rho}{\partial t}=-V_{\max}\left(1-2\frac{\rho}{\rho_{\max}}\right)\frac{\partial \rho}{\partial x}+D\frac{\partial^2 \rho}{\partial x^2}  
\end{align}
which is a Burgers equation for a mathematical model of traffic flow based on fluid dynamics\cite{Lighthill-Whitham}\cite{Whitham}
In fact,
\begin{align*}
    \lim_{\delta \to +0} \dashint_{-\infty}^\infty \coth \frac{\pi}{2\delta}(x-y)\frac{\partial \rho}{\partial y}(y)dy
    &=\dashint_{-\infty}^\infty \sgn(x-y) \frac{\partial \rho}{\partial y}(y)dy\\
    &=2\rho(x)-\rho(\infty)-\rho(-\infty),
\end{align*}
and \eqref{eq:Satsuma-Tomoeda} turns to \eqref{eq:traffic_Burgers}.
Equation \eqref{eq:Satsuma-Tomoeda} was first investigated in Ref.\cite{Satsuma-Mimura} in an equivalent form as
\begin{align*}
    \frac{\partial u}{\partial t}=-\frac{\partial}{\partial x}\left(T(u)\cdot u \right)+D\pdv[2]{u}{x},
\end{align*}
where one soliton solution has been obtained by Hirota's bilinear method and extension to periodic solutions have been discussed.

Motivated by the application of the singular integral equation to traffic flow problems, this paper firstly considers integrable integro-differential equations and their general solutions as reductions from the well-established equations.
It is based on analytic properties of the solutions in a complex domain using fundamental residue formula. 
Although this approach is essentially equivalent to that with the RH problem, we need not consider the compatibility condition between the evolution equation and the constraint on the solutions appearing in the RH problem, because the solution of integrable differential or difference equations in the complex domain can be obtained directly, and it turns into a solution in the range of the real axis of the corresponding intego-differential equation.
As a natural extension of our approach, we show that the singular integral $T\left(v(x,t))\right)$ can be generalized to that with elliptic function.
Then, we apply the method to the traffic flow model \eqref{eq:Satsuma-Tomoeda} and discuss the interaction between free-flow and congested regions which depend on the parameter of non-locality.
The structure of this paper is as follows:
First, in Section~\ref{sect:2_residue}, we discuss the integral representation of holomorphic functions and their relationship with boundary values. Using this correspondence, we clarify the connection between systems in a domain and those on the boundary, leading us to naturally observe the emergence of singular integrals. In Section~\ref{sect:Generalization}, we tackle specific examples. We extend representative integrable systems, such as the KdV equation and the Toda equation, to equations that incorporate singular integrals, then, using these approach, we examine the traffic flow model in Section~\ref{sect:traffic_flow}. In the final section, we offer concluding remarks.

\section{Singular integral and residue formula}\label{sect:2_residue}
In this section, we apply residue formula to establish the relation between singular integral and difference operator for an analytic complex functions.
Hereafter we often represent partial derivatives using subscripts and omit the dependent variables $(x,t)$. 
For instance, \(\frac{\partial^2 u}{\partial x^2}(x,t)\) is denoted as \(u_{2x}\).
As an example, let us consider the integro-differential equation 
\begin{align}
&u_t+u_{3x}-6T(u_x)u_x = 0,
\label{eq:singular-KdV}
\end{align}
where \( T \) denotes a singular integral operator, 
and it acts on a function \( v \) as 
\begin{align}\label{def:singular_T}
Tv (x):= \frac{-1}{2\delta}\dashint_{-\infty}^\infty \coth{\frac{\pi}{2\delta}(x-x')}v(x')\,dx'.
\end{align}
This equation is an extension of a nonlinear partial differential equation with a dispersion term to have nonlinear non-local effects, and 
can be obtained through Riemann-Hilbert problem related to potential KdV equation \cite{Santini_review}. 
To show effectiveness of our approach, though it is simple, we examine \eqref{eq:singular-KdV} in some detail.
\subsection{Equations obtained in the limit of the parameter $\delta$}
We show that Eq.~\eqref{eq:singular-KdV} reduces to the KdV equation in the limit $\delta \rightarrow +0$.
Let us rescale $u$ as \(\Tilde{u}(x,t):=u(x,t)/\delta \). 
We denote \( \Tilde{u} \) as \( u \) again, and impose   
the boundary condition $u\,\rightarrow \, 0 \quad \text{as} \quad |x| \rightarrow \infty$.
By taking the limit \( \delta \rightarrow +0 \),  
the nonlinear term of \eqref{eq:singular-KdV} becomes 
\begin{align*}
    -6\delta(Tu_x)u_x &\rightarrow 3\left(\int_{-\infty}^x u_x(x')\,dx' - \int_{x}^\infty u(x')_x\,dx'\right)u_x \\
    & = 3(u(x) - u(-\infty) - u(\infty) + u(x))u_x \\
    & = 6uu_x,
\end{align*}
where we used the relation \( \coth\frac{\pi}{2\delta}x \rightarrow \text{sgn}(x) \). 
Therefore, in the limit \( \delta \rightarrow +0 \), it reduces to the KdV equation.
On the other hand, in the limit \( \delta \rightarrow \infty \), 
we have
\begin{align}
    \frac{-1}{2\delta}\coth\frac{\pi}{2\delta}x &= \frac{-1}{\pi}\frac{1 + \mathcal{O}\left(\frac{1}{\delta^2}\right)}{x + \mathcal{O}\left(\frac{1}{\delta^3}\right)} \\
    &\rightarrow \frac{-1}{\pi}\frac{1}{x}.
\end{align}
Using the Hilbert transform defined by 
\begin{align}
    \mathcal{H}v(x) := \frac{1}{\pi} \dashint_{-\infty}^\infty \frac{v(x')\, dx'}{x-x'},
\end{align}
the equation reduces to
\begin{align}
    u_t + u_{3x} + 6(\mathcal{H}u_x)u_x = 0.
\end{align} 

Since both limits give integrable equations, we expect that \eqref{eq:singular-KdV} is also an integrable equation as the ILW equation \eqref{eq:ILW1}. In fact, we shall see in the following subsections that it is obtained from the KdV equation on a complex plane by imposing a solution to be holomorphic in a stripe domain.  
\subsection{Holomorphic function in a complex domain and its boundary values}
\begin{prop}\label{prop:cosh_residue}
We consider the complex domain \( D := \{z| -\delta < \Im z < \delta\} \).
Let \( U \) be holomorphic function in \( D \) and H\"{o}lder continuous on \( \partial D \). 
We denote $z=x+\ii y$ $(x,\,y \in \mathbb{R})$ and assume that
\begin{align}
\lim_{{x \to \pm \infty}} U(x + \ii y) &= U_{\pm\infty} \quad (\forall y),
\end{align}
where \( U_{\pm\infty} \) are constants. 
Then the following relation holds.
\begin{align}
U(x + \ii \delta)+U(x-\ii\delta)-U_\infty-U_{-\infty} &= -\frac{1}{\ii}T\left( U(x + \ii \delta)-U(x-\ii\delta)\right)
\end{align}
\end{prop}
\begin{proof}
Let us consider the following complex integral:
\begin{equation}\label{singular_integral}
\oint_{C}\frac{1}{2\delta}\tanh\left[\frac{\pi}{2\delta}(x-z)\right]\,f(z)\,dz
\end{equation}
The integration path is shown in Figure \ref{fig:contour}.
\begin{figure}[htbp]
\centering
\includegraphics[scale=0.5]{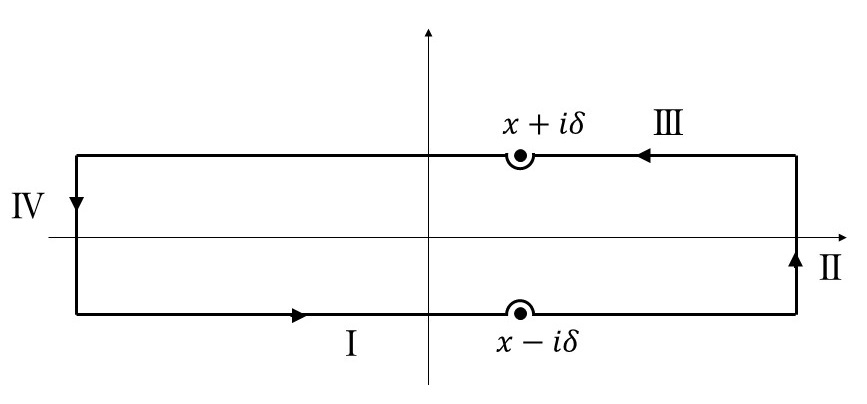}
\caption{Integration path \(C\)}
\label{fig:contour}
\end{figure}
Let the constant \( \delta > 0 \) and assume that \( U(z) \) has no singularities on \( \Im z= \pm \delta \) and the singularities inside the contour are only poles.

Given that
\begin{align}
\tanh \frac{\pi}{2\delta}(x-y+\ii \delta) = \coth \frac{\pi}{2 \delta}(x-y)
\end{align}
and near \( x-y \sim 0 \), \( \coth \frac{\pi}{2 \delta}(x-y) \sim \frac{2 \delta}{\pi (x-y)} \), we deduce

\begin{align}
\int_I &\rightarrow \dashint_{-\infty}^\infty \coth \left[\frac{\pi}{2 \delta}(x-y)\right]\, U(y-\ii \delta)\, dy +2\ii\delta U(x-\ii \delta) \\
\int_{I\!I\!I} &\rightarrow -\dashint_{-\infty}^\infty \coth \left[ \frac{\pi}{2 \delta}(x-y) \right]\, U(y+\ii \delta)\, dy +2 \ii \delta U(x+\ii \delta) 
\end{align}

From the assumption, we have
\begin{equation}
\int_{I\!I}\ \rightarrow \ -2\ii \delta U_\infty,\qquad\qquad  
\int_{I\!V}\ \rightarrow \ -2\ii \delta U_{-\infty}
\end{equation}



Thus, we obtain

\begin{align}
&\dashint_{-\infty}^\infty \coth \left[\frac{\pi}{2 \delta}(x-y)\right]\, \left\{U(y-\ii \delta) -U(y+\ii \delta)\right\}\, dy \nonumber\\
&=-2\ii \delta \left\{U(x+\ii \delta)+U(x-\ii \delta) -U_{\infty}-U_{-\infty}\right\} \label{eq:residue_formulae}
\end{align}

which completes the proof.
\end{proof}

The following corollary immediately follows from prop~\ref{prop:cosh_residue}.
\begin{cor}\label{cor:cosh_residue}
When
   \begin{align}
&U(x - \ii \delta) = U^*(x + \ii \delta) \\
&U_\infty+U_{-\infty}=0
\end{align}
holds, writing \( U(x + \ii \delta) \) = : \( w(x) + \ii v(x) \),
we have
\begin{align}
w(x) = \dashint_{-\infty}^\infty \frac{-1}{2\delta} \coth \left[\frac{\pi}{2 \delta}(x-y)\right]\,  v(y)\, dy  \label{eq_im-re-relation}
\end{align}
\end{cor}

\subsection{Derivation of Eq.~\eqref{eq:singular-KdV}}
We show that Eq.~(\ref{eq:singular-KdV}) is derived from the complex KdV equation
\begin{align}\label{eq:Complex_KdV}
    U(z,t)_t + U(z,t)_{3z} + 6U(z,t)U(z,t)_z = 0. 
\end{align}
on the domain \( D \).
Assume that \( U(z,t) \) is a solution of the KdV equation on \( D \) that satisfies the assumptions of corollary~\ref{cor:cosh_residue}.
Taking the limit of \( U \) as \( \Im z \rightarrow -\delta \) and using the real function \( v(x,t) \), we have
\begin{align}
  U(x+\ii\delta) = -Tv(x,t) + \ii v(x,t)  
\end{align}
which satisfies
\begin{align}
    (-Tv + \ii v)_t + (-Tv + \ii v)_{3x} + 6(-Tv + \ii v)(-Tv + \ii v)_x = 0.
\end{align}
Therefore, by extracting the imaginary part, we obtain
\begin{align}
    v_t + v_{3x} - 6(v Tv)_x = 0.
\end{align}
By substituting \( v = u_x \) and integrating once with respect to \( x \), it yields (\ref{eq:singular-KdV}).

\subsection{Soliton solutions}
Once the origin of the singular integral equation \eqref{eq:singular-KdV} is clarified as above, its properties such as general solutions, Lax pairs, and conserved quantities can be straightforwardly  obtained from those of the KdV equation. 

In this subsection, we discuss the solutions for (\ref{eq:singular-KdV}).
There are lots of established methods to construct the solutions of the KdV equation,  such as inverse scattering methods, Hirota bilinear methods, reduction from the KP hierarchy. 
Here we consider $N$ soliton solutions using the tau function $\tau(z,t)$.
The solution to the KdV equation is written using the tau function $\tau(z,t)$ as
\begin{align}
U(z,t) &= 2 \frac{\partial^2}{\partial z^2} \log \tau(z,t)
\end{align}
The $N$ soliton solutions are given as
\begin{align}
\tau(z,t)&=\sum_{\{\mu_k\}_{k=1}^N\in\{0,1\}^N}\exp{\sum_{i=1}^N\mu_i\eta_i +\sum_{1\le i<j\le N}\mu_i\mu_jA_{ij} },
\end{align}
where 
\begin{align*}
    \eta_i&:=k_i z-k_i^3 t-\theta_i,\quad (i=1,2,...,N)\\
    \e^{A_{ij}}&:=\frac{(k_i-k_j)^2}{(k_i+k_j)^2}\quad (1\le i<j\le N),
\end{align*}
for real parameters $k_i$ ($k_i \neq\pm k_j$ for $i\ne j$) and $\theta_i$. 
To ensure that the solution is holomorphic, it is sufficient to set \(0 < k_j\delta < \pi\) for any \(j\).

Similar to the derivation of the equation, 
by taking a limit as  \( \Im z \rightarrow \delta \), 
taking its imaginary part, 
and integrating once with respect to \( x \), 
the solution to (\ref{eq:singular-KdV}) is expressed as
\begin{align}\label{sol:singular-KdV}
u &=  2\frac{\Re \tau (\Im \tau)_x - (\Re \tau)_x \Im \tau}{\Re \tau^2 + \Im \tau^2}.
\end{align}
%

Specifically, we provide explicit expressions for the 1-soliton and 2-soliton solutions.
For a 1-soliton, the \(\tau\) function corresponding to \( U \) is given by
\begin{align}
\tau_1 = 1 + \e^{kz - k^3 t}
\end{align}
where \( k \in \mathbb{R} \) and we set $\theta=0$.
From the derived formula (\ref{sol:singular-KdV})
we find that the solution is
\begin{align}
u = \frac{k\sin{k\delta}}{\cosh{(kx-k^3t)}+\cos{k\delta}}.
\end{align}
Similarly, for a 2-soliton solution, 
the \(\tau\) function that provides the solution \( U \) is given by
\begin{align}
\tau_2 = 1 + \e^{k_1z - k_1^3 t} + \e^{k_2z - k_2^3 t} + \e^{(k_1 + k_2)z -(k_1^3 + k_2^3) t+A_{12}}
\end{align}
where \( k_1 , k_2 \in \mathbb{R} \) and \( \e^{A_{12}} = \left(\frac{k_1-k_2}{k_1+k_2}\right)^2 \). Instead of giving the concrete expression of the corresponding solution, let us discuss the asymptotic behavior of the solution corresponding to the 2-soliton solution as \( t \) approaches \( \pm \infty \).
We assume \( 0 < k_1 < k_2 \). 
Firstly let us consider a region $x-k_1^2t=const.$, that is, we consider a region where a soliton with wavelength $k_1$ exists. In this region, 
since $\e^{\eta_2} \gg \e^{\eta_1}$ as $t \to -\infty$,
we have
\begin{align}\label{eq:tau_2_2_plus}
\begin{cases}
        \Re \tau_2 &\sim \e^{\eta_2}(\cos k_2\delta + a_{12} \e^{\eta_1}\cos(k_1+k_2)\delta),\\
        \Im \tau_2 &\sim \e^{\eta_2}(\sin k_2\delta + a_{12} \e^{\eta_1}\sin(k_1+k_2)\delta),\\
        {\Re \tau_2}_x &\sim \e^{\eta_2}(k_2\cos k_2\delta + (k_1+k_2)a_{12} \e^{\eta_1}\cos(k_1+k_2)\delta),\\
        {\Im \tau_2}_x &\sim \e^{\eta_2}(k_2\sin k_2\delta + (k_1+k_2)a_{12} \e^{\eta_1}\sin(k_1+k_2)\delta).
\end{cases}
\end{align}
From (\ref{eq:tau_2_2_plus}),
we obtain
\begin{align}
    u \sim \frac{k_1\sin k_1\delta}{\cosh(\eta_1 + \alpha) + \cos k_1\delta},
\end{align}
where \(\alpha := 2\log \frac{k_2-k_1}{k_2 + k_1}\). 

As \(t \rightarrow \infty\), $\e^{\eta_1}\gg \e^{\eta_2}$ and \(\e^{\eta_2} \rightarrow 0\), we find
\begin{align}\label{eq:tau_2_2_minus}
\begin{cases}
        \Re \tau_2 &\sim 1 + \e^{\eta_1} \cos k_1\delta,\\
        \Im \tau_2 &\sim \e^{\eta_1} \sin k_1\delta,\\
        {\Re \tau_2}_x &\sim k_1\e^{\eta_1} \cos k_1\delta,\\
        {\Im \tau_2}_x &\sim k_1\e^{\eta_1} \sin k_1\delta.
\end{cases}
\end{align}
From (\ref{eq:tau_2_2_minus}),
we obtain
\begin{align}
    u \sim \frac{k_1\sin k_1\delta}{\cosh \eta_1 + \cos k_1\delta}.
\end{align}
Similarly, in the region where a soliton with wavenumber $k_2$ exists,
as \( t \rightarrow -\infty \), 
\begin{align}
    u \sim \frac{k_2\sin k_2\delta}{\cosh \eta_2 + \cos k_2\delta}
\end{align}
and as \( t \rightarrow \infty \), 
\begin{align}
    u \sim \frac{-k_2\sin k_2\delta}{\cosh(\eta_2 + \alpha) + \cos k_2\delta}.
\end{align}
Therefore, the asymptotic behavior of the 2 soliton solution is as follows:
\begin{align}
    \begin{cases}
        u \sim \dfrac{k_1\sin k_1\delta}{\cosh(\eta_1 + \alpha) + \cos k_1\delta} + \dfrac{k_2\sin k_2\delta}{\cosh \eta_2 + \cos k_2\delta} \quad (t \rightarrow -\infty),\\[10pt]
        u \sim \dfrac{k_1\sin k_1\delta}{\cosh \eta_1 + \cos k_1\delta} + \dfrac{k_2\sin k_2\delta}{\cosh(\eta_2 + \alpha) + \cos k_2\delta} \quad (t \rightarrow \infty).
    \end{cases}
\end{align}
%


\subsection{Lax equation and conserved quantities}
The Lax equation associated with the KdV equation is 
\begin{align}
    \pdv{L}{t} = [B,L] \label{eq:complex Lax}
\end{align}
for
\begin{align}
    \begin{cases}
        L &= - \partial_z^2 + U, \\
        B &= -4\partial_z^3 + 6U\partial_z + 3U_z.
    \end{cases}
\end{align}
It is obvious that the Lax equation for the singular integral equation \eqref{eq:singular-KdV} is equivalent to \eqref{eq:complex Lax}.
Since $U(x+\ii\delta,t)=-T(v)(x,t)+\ii v(x,t)$, by considering  the real  and imaginary parts of the operator:
\begin{dfn}
\begin{align}
    \begin{cases}
       L_R &:= -\partial^2 - T(v),\\
        L_I &:= v,
    \end{cases}
\end{align}
and
\begin{align}
    \begin{cases}
    B_R&:= -4\partial^3 - 6T(v)\partial - 3(T(v))_x,\\
    B_I&:= 6v\partial + 3v_x,
    \end{cases}
\end{align}
\end{dfn}
we find
\begin{align}
        \pdv{}{t}L_R &= [B_R, L_R] - [B_I, L_I],\label{eq:real lax kdv}\\
        \pdv{}{t}L_I &= [B_I, L_R] + [B_R, L_I].\label{eq:im lax kdv}
\end{align}
Computing (\ref{eq:real lax kdv}) and (\ref{eq:im lax kdv}) respectively, we get 
\begin{align}
    (T(v))_t = -  (T(v))_{3x} + 3( (T(v))^2 + v^2)_x\label{eq:real lax explicit kdv}, 
\end{align}
and
\begin{align}
    v_t = -v_{3x} + 6( (T(v))\cdot v)_x. \label{eq:im lax explicit kdv}
\end{align}
Clearly \eqref{eq:im lax explicit kdv} is equivalent to \eqref{eq:singular-KdV}.
At a glance, the Lax representation for \eqref{eq:singular-KdV} consists of two simultaneous equations \eqref{eq:real lax kdv} and \eqref{eq:im lax kdv}, because of holomorphic nature of the dependent variable $U$, it is shown that if one holds, the other also holds.
\begin{prop}
    If either (\ref{eq:real lax explicit kdv}) or (\ref{eq:im lax explicit kdv}) holds, then the other also holds.
\end{prop}
\begin{proof}
    We demonstrate that when (\ref{eq:im lax explicit kdv}) holds, 
(\ref{eq:real lax explicit kdv}) also holds. The converse can be proven in a similar manner. 

Let $a:=\Re U(=-T(v))$ and $b:=\Im U(=v)$ for simplicity.
By performing partial differentiation with respect to \( y \) on both sides of equation (\ref{eq:real lax explicit kdv}) and substituting in the Cauchy-Riemann equations \( a_x = b_y \) and \( a_y = -b_x \), we obtain 
\begin{align}
    a_{xt} = -a_{4x} - 6(aa_x - bb_x).
\end{align}
By integrating once with respect to \(x\), we obtain (\ref{eq:real lax explicit kdv}).
\end{proof}
%
%

Conserved quantities of (\ref{sol:singular-KdV}) is also constructed from the KdV equation \eqref{eq:Complex_KdV}. 
By applying the generalized Gardiner transformation 
\begin{align}\label{eq:Gardiner_trans}
    U = -\{W + \varepsilon W_z + \varepsilon^2 W^2\}
\end{align}
to (\ref{eq:Complex_KdV}), we get
\begin{align}
U_t + U_{3z} + 6UU_z = -\left(1 + \varepsilon\pdv{}{z} + 2\varepsilon^2 W \right)\left(W_t -6(W+\varepsilon^2W^2)W_z + W_{3z}\right).
\end{align}
This implies that if the condition 
\begin{align}\label{eq:cons_relation}
    W_t -6(W+\varepsilon^2W^2)W_z + W_{3z}=W_t +(-3W^2-2\varepsilon^2W^3+W_{2z})_z=0
\end{align}
holds, then \( U \) is a solution to the KdV equation.
Because a conserved quantity is a pair $(\rho, J)$, where both $\rho$ and $J$ are polynomials of $U$ and its partial derivatives of $z$, and satisfy
\begin{align*}
    \rho_t+J_z=0,
\end{align*}
or $\left(\int \rho \, dx\right)_t=0$ for $J(+\infty,t)-J(-\infty,t)=0$,
\eqref{eq:cons_relation} gives an infinite number of conserved quantities. 
We assume 
\begin{align}
    W = W_0 + \varepsilon W_1 + \varepsilon^2 W_2 + \varepsilon^3 W_3 + \varepsilon^4 W_4 + \varepsilon^5 W_5 + \cdots 
\end{align}
and solve (\ref{eq:Gardiner_trans}) sequentially, we get
\begin{align}
    &W_0 = -U,\  W_1 = U_z,\  W_2 = -U_{2z} - U^2,\nonumber \\
    &W_3 = U_{3z}+4UU_z,\  W_4 =  - U_{4z}-5U_z^2-6UU_z -2 U^3,\ \cdots\nonumber.
\end{align}
From (\ref{eq:cons_relation}), \( W \) is a conserved quantity, and since \( \varepsilon \) is arbitrary, \( W_n \) (\( n = 0, 1, 2, \ldots \)) are also conserved quantities.
Note that \( W_n \) is a differential polynomial of \( U \). 
Therefore, if \( U \) is holomorphic, \( W_n \) is also holomorphic.
Furthermore, non-trivial conserved quantities exist for even \( n \).
Thus, for \( W_n \), by taking the imaginary part as \( \Im z \rightarrow \delta \), 
we find that the infinite number of conserved quantities $\{\rho_k\}_{k=0}^\infty$ are given as
\begin{align}
&\rho_0 = u_x, \quad \rho_1=u_{xx},\quad \rho_2 = u_{3x}+2u_xTu_x, \nonumber\\
&\rho_3=u_{4x}-4(u_xTu_x)_x,\nonumber\\
&\rho_4=-u_{5x}+10u_{xx}Tu_{xx}+6(u_xTu_x)_x -2(3(Tu_x)^2u_x-u^3_x),\cdots .
\end{align}
%
%
%
\subsection{Another example related to (\ref{eq:singular-KdV})}
The above mentioned analytic reduction gives different equations which incorporate several integral terms.
An example is 
\begin{align}\label{eq:singular_KdV_it}
    Tu_t + Tu_{3x} - 3\left((Tu)^2+ u^2\right)_x = 0.
\end{align}
This equation and solutions are obtained from Eq.\eqref{eq:Complex_KdV} and just the counterpart of (\ref{eq:singular-KdV}).
The solution is obtained from that of the complex KdV equation as
\begin{align*}
    u(x,t) = \Re U(x+\ii\delta,\ii t).
\end{align*}
A solution which corresponds to a 1-soliton solution is 
\begin{align}
    u(x,t) &=\Re\left[\frac{k^2}{(1+\e^{kx+x^3t+\ii k \delta})(1+\e^{-kx-k^3t-\ii k \delta})} \right],\nonumber \\
    &=\frac{-k^2\{1+\cosh (kx+k^3t) \cos k \delta\}}{\cosh^2(kx+k^3t)+2\cosh (kx+k^3 t)\cos k\delta +\cos^2k\delta},
\end{align}
where \(k\) is a real constant.

%
%
\section{Generalization of integrable integro-differential equations}\label{sect:Generalization}
%
In this section, we extend the construction of integrable integro-differential equations shown in the previous section, and obtain, hierarchies, difference-integral equations, elliptic singular integral equations. 

\subsection{Generalization to the KP hierarchy}

Application of the settings in section~\ref{sect:2_residue} to the KP hierarchy is straight forward and gives a series of integrable integro-differential equations.
In fact, what we have to consider is if the solutions are holomorphic in the given domain or not.
Let us recall the construction of the KP hierarchy\cite{Miwa-Jimbo-Datetextbook}.
We consider the pseudo-differential operator which is sometimes called the dressing operator:
\begin{align*}
    W(\partial^{-1}):=1+w_1(\x)\partial^{-1}+w_2(\x\partial^{-2}+...,
\end{align*}
where $\x:=(x_1\equiv z,x_2,x_3,...)$ denotes an infinite number of independent variables, $w_k(\x)$ ($k=1,2,...$) are infinite number of dependent variables, $\partial\equiv \frac{\partial}{\partial z}$, and $\partial^{-1}$ is its formal inverse that satisfies
\begin{align*}
    \partial^{-n}f=\sum_{k=1}^\infty \begin{pmatrix}-n \\ k \end{pmatrix}f^{(k)}\partial^{-n-k}, \quad 
    \begin{pmatrix}-n \\ k \end{pmatrix}=\frac{(-n)(-n-1)\cdots(-n-k+1)}{k!}.
\end{align*}.
Then, denoting the differential part of a pseudo-differential operator $P$ by $(P)_+$,
\begin{align*}
    L&:=W \partial W^{-1},\\
   B_n&:=(W\partial^nW^{-1})_+,
\end{align*}
we obtain the KP hierarchy
\begin{align*}
    (L)_{t_n}=\left[B_n, L \right]\qquad (n=2,3,4,...)
\end{align*}
and the Zakharov-Shabat equations for ${}^\forall k,\, l$: 
\begin{align}
    \left(\partial_k \hat{B}_l\right)-\left(\partial_l \hat{B}_k\right)=\left[\hat{B}_k, \hat{B}_l\right]\equiv
    \hat{B}_k\hat{B}_l-\hat{B}_l\hat{B}_k. \label{eqn_Zakharov_Shabat}.
\end{align}

The solutions to the KP hierarchy and the Zakharov-Shabat equations are given by the tau function $\tau(\x)$.
For $w_k(\x)=0$ ($k \ge n+1$), $\tau(\x)$ is given by the Wronskii determinant
\begin{align}\label{tau:Wronskii}
    \tau(\x):=\mdet{f_1&f_2&\cdots&f_n\\(f_1)_z&(f_2)_z&\cdots&(f_n)_z\\
    \vdots&\vdots&\ddots&\vdots\\(f_1)_{(n-1)z}&(f_2)_{(n-1)z}&\cdots&(f_n)_{(n-1)z}}.
\end{align}
Here $f_i(\x)$ ($i=1,2,...,n$) are $n$ independent functions which satisfy simultaneous linear partial differential equations:
\begin{align*}
    (f)_{x_k}=\partial^k f \qquad (k=2,3,4,...).
\end{align*}
By putting 
\begin{align*}
\hat{W}:=W\partial^n=\partial^n+w_1(\x)\partial^{n-1}+\ldots+w_n,    
\end{align*}
$w_k(\x)$ ($k=1,2,...,n$) are determined by
\begin{align*}
    \hat{W}f_j(\x)=0\qquad (j=1,2,..,n)
\end{align*}
and expressed by $\tau(\x)$ and its derivatives with respect to $z$.
For example $w_1(\x)=(\log \tau(\x))_z$.

To obtain series of integro-differential equations, we suppose the constraint with $\partial^*\equiv \frac{\partial}{\partial z^*}$:
\begin{align*}
   \left[\partial^*, W(\partial^{-1})\right]=0,\qquad  W(\partial^{-1})^*=W(\partial^{-1}).
\end{align*}
in some domain $D \subseteq \bC$.
This means that 
\begin{align*}
    \partial^*w_i(\x)=0\quad (i=1,2,...) \quad (z\in D).
\end{align*}
Applying the residue formula \eqref{eq:residue_formulae} with appropriate boundary conditions, 
\begin{align*}
&\dashint_{-\infty}^\infty \coth \left[\frac{\pi}{2 \delta}(x-y)\right]\, \left\{w_i(y-\ii \delta) -w_i(y+\ii \delta)\right\}\, dy \nonumber\\
&=-2\ii \delta \left\{w_i(x+\ii \delta)+w_i(x-\ii \delta) \right\} 
\quad (i=1,2,...,n),
\end{align*}
where $w_i(x+\ii\delta)$ is an abbreviation of $w_i(x+\ii\delta,x_2,x_3,...)$.
If we further suppose that $w_j(\x)$ is analytic with respect to $x_k$, we have similar relation and obtain singular integral equations for $x_k$. 

As an example, let us consider the simplest Zakharov-Shabat equation for $k=2,\, l=3$. By putting $u(\x):=(w_1(\x))_z$ and $x_1=x,\, x_2=y,\, x_3=t$,
we find an integro-differential equation of KP-type as
\begin{align}\label{eq:singular_KP_x}
    (4u_t +6(uT^{(\alpha)}u)_x - u_{3x})_x = 3u_{yy}\quad (\alpha=x,\,y)
\end{align}
for \(u = u(x,y,t)\).
Here $T^{(x)}$ ($T^{(y)}$) is the singular integral operator \eqref{def:singular_T} with respect to $x$ ($y$).
To derive Eq. \eqref{eq:singular_KP_x} and solutions of that we consider the following Kadomtsev-Petviashvili (KP) equation on \( D = \{z | -\delta_1  < \Im z < \delta_1\} \times \{w | -\delta_2 < \Im w < \delta_2\}\) for \( U = U(z, w, t) \):
\begin{align}\label{eq:complex_KP}
    (U_t - 6UU_z - U_{3z})_z = 3U_{ww}.
\end{align}
Let \(U(z,w,t)\) be a holomorphic solution for \eqref{eq:singular_KP_x} which means \(U\) is holomorphic in each variable \(z\) or \(w\).
Focusing on the boundary value as \(\Im z \to -\delta_1\) or \(\Im w \to -\delta_2\) gives
\begin{align*}
    U(x+\ii\delta_1,y,t) &= -T^{(x)}u(x,y,t) + \ii u(x,y,t),\\
    U(x,y+\ii\delta_2,t) &= -T^{(y)}u(x,y,t) + \ii u(x,y,t),
\end{align*}
where \(u(x,w,t)\) is a suitable real function.
Inserting this into Eq.\eqref{eq:complex_KP} and extracting the imaginary part yields Eq.\eqref{eq:singular_KP_x}. 
Furthermore, it becomes apparent that the solution is given by \( u(x,y,t) = \Im U(x-\ii\delta_1,y,t) \) or \( u(x,y,t) = \Im U(x,y-\ii\delta_2,t) \).\\
For example, to compute a solution corresponds to 1-soliton solution,
we take 
\begin{align}
    \tau(z,w,t) = 1 + \e^{kz + k^2w + k^3t},
\end{align}
where \(k\) is real.
Since a solution of Eq.\eqref{eq:complex_KP} is given by
\begin{align}
    U(z,w,t) = 2\pdv[2]{}{z}\log \tau(z,w,t),
\end{align}
we obtain
\begin{align}
    u(x,y,t) = -k^2 \frac{\sinh(kx + k^2y + k^3t)\sin(k\delta_1)}{\left(\cosh^2(kx + k^2y + k^3t)+\cos(k\delta_1)\right)^2},
\end{align}
or
\begin{align}
    u(x,y,t) = -k^2 \frac{\sinh(kx + k^2y + k^3t)\sin(k^2\delta_2)}{\left(\cosh^2(kx + k^2y + k^3t)+\cos(k^2\delta_2)\right)^2},
\end{align}
The sufficient condition on which \(\tau\) is holomorphic is that the inequality 
\begin{align}
    -\pi < k\delta_1< \pi, \ \mbox{or} \  -\pi <  k^2\delta_2 < \pi
\end{align}
holds.
%
\subsection{ILW equation}
Here, we briefly argue the long wave (ILW) equation \eqref{eq:ILW1}.
Using the identity \eqref{eq:residue_formulae}, when $u(x)$ in \eqref{eq:ILW1} is expressed as
\begin{align}
    u(x)=-\frac{\ii}{2}\left[U(x+\ii \delta )-U(x-\ii\delta)\right], \quad U_\infty+U_{-\infty}=0,
\end{align}
we have 
\begin{align}\label{eqn:bilinear_ILW}
(U^+-U^-)_t+\frac{1}{\delta}(U^+-U^-)_x-\ii(U^+-U^-)(U^+-U^-)_x-\ii(U^++U^-)_{xx}=0,    
\end{align}
where $U^+:=U(t,x+\ii\delta)$ and $U^-:=U(t,x-\ii\delta)$.
By substituting $U^\pm=(\log F^\pm)_z$, we find
\begin{align}\label{eq:modified_KdV}
    \left[D_t+\frac{1}{\delta}D_x+\ii D_x^2+C(t) \right] F^+\cdot F^-=0,
\end{align}
where $D_t$ and $D_x$ are Hirota bilinear operators ($D_x^m a\cdot b =\lim_{x' \to x} (\partial_x-\partial_{x'})^m a(x)b(x')$ etc.), and $C(t)$ is an arbitrary smooth function of $t$.
Equation \eqref{eq:modified_KdV} is the lowest bilinear equation of the modified KdV hierarchy and its solutions and generalisation have been discussed in \cite{TutiyaSatsuma2003} in detail.

\subsection{Toda type equation}
Next, we turn our attention to the Toda lattice, which stands as a paradigmatic integrable system. 
The current-voltage form of the Toda lattice for \(I_n(t), V_n(t)\) is represented as:
\begin{align}
\begin{cases}
   \displaystyle \dv[]{I_n(t)}{t} &= V_{n-1}(t) - V_n(t), \\[5pt]
   \displaystyle  \dv[]{V_n(t)}{t} &= V_n(t)(I_n(t) - I_{n+1}(t))\label{eq:toda-vi-form},
\end{cases}
\end{align}
where \( n \in \mathbb{Z}, t \in \mathbb{R}\). 
Traditionally, in \(I_n(t), V_n(t)\), \(n\) denotes a discrete spatial variable, while \(t\) signifies a continuous time variable. 
For the purpose of streamlining our interpretation of the equations here, 
we regard \(n\) as a discrete time variable and rewrite \(t\) as \(x\), interpreting it as a continuous spatial variable.

We perform operations similar to those in the previous section.
For the holomorphic solutions \(I_n(z)\) and \(V_n(z)\), 
when taking the limit as \(\Im z \rightarrow \delta\), 
using the real functions \(u_n(x)\) and \(v_n(x)\), 
we have
\begin{align}
\begin{cases}
I_n(x+\ii \delta) &= -T u_n(x) + iu_n(x) \\
V_n(x + \ii \delta) &= -T v_n(x) + iv_n(x).
\end{cases}
\end{align}

By Substituting the relation into equation (\ref{eq:toda-vi-form}) 
and taking the imaginary part, 
then introducing \(w_n := u_{n+1} - u_n\), 
we obtain
\begin{align}\label{eq:singular-Toda}
\begin{cases}
\displaystyle\frac{-d w_n(x)}{dx}&=v_{n+1}(x) - 2v_n(x) + v_{n-1}(x), \\[5pt]
\displaystyle\frac{d v_n(x)}{dx} &= v_n(x)Tw_n(x) + w_n(x) Tv_n(x).
\end{cases}
\end{align}

The tau function corresponding to the 1-soliton solution of the Toda lattice is given by 
\begin{align}
    \tau_n(z) = 1 + \e^{kn + lz},
\end{align}
where \(l\) satisfies \( l = \pm 2 \sinh{\frac{k}{2}}\).
Through calculation, we find
\begin{align}
    \begin{cases}
        w_n &= \dfrac{l \sin{l\delta}}{2}\left(\dfrac{1}{\cosh{(k(n+1) + lx)}+ \cos{l\omega}}-2\dfrac{1}{\cosh({kn + lx})+ \cos{l\omega}} \right.\\
        &\qquad \qquad \qquad \qquad \left.+ \dfrac{1}{\cosh{(k(n-1) + lx)}+ \cos{l\omega}} \right),\\[10pt]
        v_n &= 2 \sin l \delta \sinh^2(\dfrac{k}{2}) \dfrac{\sinh(kn+lx)}{(\cosh(kn + lx) + \cos l\delta)^2}.
    \end{cases}
\end{align}

The tau function associated with 2-soliton is
\begin{align}
    \tau_n(z) = 1 + \e^{k_1n + l_1z} + \e^{k_2n + l_2z} + a_{12}\e^{(k_1+k_2)n + (l_1+l_2)z},
\end{align}
where 
\begin{align*}
&l_i= \pm 2\sinh^2\frac{k_i}{2} \quad(i = 1,2), \qquad a_{12}=\left(\frac{\sinh \frac{k_1-k_2}{4}}{\sinh \frac{k_1+k_2}{4}}\right)^2.
\end{align*}

\subsection{Elliptic singular integral}
We extend the singular integral \eqref{def:singular_T} to an elliptic singular integral.
The Jacobi sn function of modulus $k$ ($0<k<1$),  $\sn(z,k)$, is defined by
the integral:
\begin{align}\label{eq:sn_function}
 z=\int_0^{\sn(z,k)}\frac{dx}{\sqrt{(1-x^2)(1-k^2x^2)}}.    
\end{align}
It is a doubly periodic function with periods $4K$ and $2\ii K'$ and has two simple poles at $\ii K'$ and $2K+\ii K'$, where
\begin{align}
   K&:=\int_0^1\frac{dx}{\sqrt{(1-x^2)(1-k^2x^2)}},\\
   K'&:= \int_0^1\frac{dx}{\sqrt{(1-x^2)(1-{k'}^2x^2)}}\qquad (k':=\sqrt{1-k^2}).
\end{align}
Its properties are summarised as 
\begin{subequations}
\begin{align}
&\sn(z+4K,k)=\sn(z+2\ii K',k)=\sn(z,k)\label{eq:ellip_a}\\
&\sn(z,k)=-\sn(-z,k),\quad \sn(z+2K,k)=-\sn(z,k) \label{eq:ellip_b}\\
&\sn(z+\ii K')=\sn(z-\ii K')=\frac{1}{k}\ns(z,k) \label{eq:ellip_c}\\
&\sn(z,k)=z-\frac{1+k^2}{3!}z^3+\frac{1+14k^2+k^4}{5!}z^5-\cdots, \label{eq:ellip_d}\\
&\Res_{z=0}\left[\ns(z,k) \right]=1,\quad \Res_{z=2K}\left[\ns(z,k) \right]=-1\label{eq:ellip_e}\\
&\lim_{k\to 1}\sn(z,k)=\tanh z,\qquad \lim_{k\to 0}\sn(z,k)=\sin z
\label{eq:ellip_f}
\end{align}
\end{subequations}
Here $\ns(z,k)$ is the Jacobi ns function: $\ns(z,k)=(\sn(z,k))^{-1}$.

Let us prove the following proposition similar to prop.~\ref{prop:cosh_residue}.
\begin{prop}\label{prop:ns_residue}
We consider the complex domain \( D := \{z| -K' < \Im z < K'\} \).
Let \( U \) be holomorphic function in \( D \) and H\"{o}lder continuous on \( \partial D \). 
Then the following relation holds.
\begin{align}
&\dashint_{x-K}^{x+K}\ns(x-y,k)\left\{U(y-\ii K')-U(y+\ii K')\right\}\,dy\nonumber\\
&+\ii \int_{-K'}^{K'}k\,\sn(-K-\ii \eta,k)\left\{U(x+K+\ii \eta)+U(x-K+\ii \eta)\right\}\, d\eta 
\nonumber\\
&=-\ii \pi \left\{U(x-\ii K')+U(x+\ii K')  \right\}\label{eq:elliptic_s_integral_formula}
\end{align}
In particular, if $U(z)$ is an anti-periodic function which satisfies
\begin{align}\label{eq:U_antisym}
    U(z+2K)=-U(z)
\end{align}
then, we have
\begin{align}
    &\dashint_{x-K}^{x+K}\ns(x-y,k)\left\{U(y-\ii K')-U(y+\ii K')\right\}\,dy\nonumber\\
&=-\ii \pi \left\{U(x-\ii K')+U(x+\ii K')  \right\}.\label{eq:elliptic_s_integral_formula2}
\end{align}
\end{prop}

The proof of proposition~\ref{prop:ns_residue} is almost the same as that of proposition~\ref{prop:cosh_residue}.
\begin{proof}
%
%
\begin{figure}[htbp]
\centering
\includegraphics[scale=0.5]{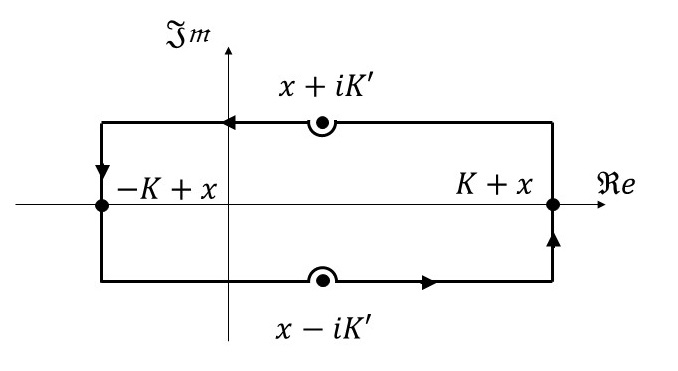}
\caption{Integration path $\Gamma$ for elliptic singular integral}
\label{fig:contour_sn}
\end{figure}
%
%
%
We consider the contour integral 
\begin{align*}
\oint_\Gamma k\,\sn(x-z,k) U(z)\, dz=0,
\end{align*}
where the contour $\Gamma$ is shown in Fig.~\ref{fig:contour_sn} and there is no singularity inside the contour its value is zero.
Then, using the properties of $\sn(z,k)$ \eqref{eq:ellip_c}, we find
\begin{align*}
    k\, \sn (y\pm \ii K',k)=\ns (y, k),
\end{align*}
and using \eqref{eq:ellip_e}, we obtain \eqref{eq:elliptic_s_integral_formula} and \eqref{eq:elliptic_s_integral_formula2}.
\end{proof}

\begin{rem}
  When $U(z^*)=U(z)^*$, by putting $U(x+\ii K')=w(x)+\ii v(x)$, we have
   \begin{align}
       w(x)=\frac{1}{\pi}\dashint_{x-K}^{x+K}\ns(x-y,k)v(y)dy+C(U)
   \end{align}
   where
   \begin{align}
       C(U)&=-\frac{1}{2\pi}\int_{-K'}^{K'}k\,\sn(-K-\ii\eta,k)\left\{U(x+K+\ii\eta)+U(x-K+\ii\eta)\right\}\, d\eta \nonumber\\
       &=\frac{1}{2\pi}\int_{-K'}^{K'}\frac{k}{\dn(\eta,k')}\left\{U(x+K+\ii\eta)+U(x-K+\ii\eta)\right\}\, d\eta.
   \end{align}
   Here we used the relation $\sn(-K-\ii\eta,k)=-\dfrac{1}{\dn(\eta,k')}$.
\end{rem}

Since $\dn (z,k')$ is an even function, similar to the corollary~\ref{cor:cosh_residue}, we immediately find the following corollary.
\begin{cor}\label{cor:elliptic_relation}
    If $U(z)$ satisfies 
    \begin{align}
        &U(z^*)=U(z)^*
    \end{align}
    then, by putting $U(x+\ii K'):=w(x)+\ii v(x)$, we have 
    \begin{align}
      w(x)&=\frac{1}{\pi}\dashint_{x-K}^{x+K}\ns(x-y,k)v(y)dy +C(x)  \\
      C(x)&:=\frac{k}{\pi}\int_0^{K'}\frac{1}{\dn(\eta,k')}\Re\left[U(x+K+\ii\eta)+U(x-K+\ii\eta)\right]
    \end{align}
    In particular, if $U(z+2K)=-U(z)$, $C(x)=0$.
\end{cor}
\begin{rem}
    In the limit $k \rightarrow 1$, 
    \begin{align*}
        K \to +\infty,\quad K'\to \frac{\pi}{2},\quad \ns(z) \to \coth z
    \end{align*}
    Hence proposition~\ref{prop:ns_residue} is a generalisation of the proposition~\ref{prop:cosh_residue}.
\end{rem}

\subsection{modified KdV-type equation with elliptic singular integral}
We can construct integrable integro-differential equations with elliptic integral in the same way as those with $\coth$ singularity.
As an example, we consider the modified KdV type integro-differential equation:
\begin{align}
    &v_t+v_{3x}+2\left\{\left(3(\tilde{T_k}v)^2-v^2\right)v\right\}_x=0 \label{eq:elliptic_mKdV}\\
    &\tilde{T_k}v(x):=\frac{1}{\pi}\dashint_{x-K}^{x+K}\ns(x-y)v(y)\,dy
\end{align}
Suppose that a complex function $U(z,t)$ is holomorphic in $-K' \le \Im{z} \le K'$ and satisfies the condition:
\begin{align}\label{eq:U_ellip_constraint}
     U(z+2K)=-U(z).
\end{align}
Note that we do not require that $U(z^*)=U(z)^*$.
If $U(z,t)$ is a solution of the complex modified KdV equation
\begin{align}\label{eq:cm_el_mKdV}
    U_t+U_{3z}-2(U^3)_z=0,
\end{align}
then, putting $v(x,t)=\frac{U(x+\ii\delta)-U(x-\ii\delta)}{2\ii}$ and 
$w(x,t)=\frac{U(x-\ii\delta)+U(x+\ii\delta)}{2}$, we find
\begin{align*}
    v_t+v_{3x}-2\left\{(3w^2-v^2)v\right\}_x=0.
\end{align*}
From Proposition~~\ref{prop:ns_residue}, $w(x,t)=\tilde{T_k}v(x,t)$ and we obtain 
\eqref{eq:elliptic_mKdV}. 

We consider a complex function $F(z,t)$ which is holomorphic and satisfies the periodic condition: 
\begin{align}
    F(z+4K,t)=F(z,t).
\end{align}
By variable transformation 
\begin{align*}
 U(z,t)=\left(\log \frac{F(z+K,t)}{F(z-K,t)}\right)_z   
\end{align*}
we find that $U(z,t)$ satisfies the constraint \eqref{eq:U_ellip_constraint}, and becomes a solution of \eqref{eq:elliptic_mKdV} when it satisfies the Hirota bilinear equation for mKdV equation \cite{Hirota1972}:
\begin{align}
 &\left(D_t+D_z^3\right)F_+\cdot F_-=0 \label{eq:mKdV_Hirota1}\\
 &D_z^2 F_+\cdot F_-=0 \label{eq:mKdV_Hirota2},
\end{align}
where $F_+:=F(z+K)$ and $F_-:=F(z-K)$.
A simple non-trivial solution of \eqref{eq:mKdV_Hirota1} and \eqref{eq:mKdV_Hirota2} is
\begin{align}
    F(z,t)&=\e^{\frac{\ii}{2} (k_m z+k_m^3 t+\phi+\ii\theta) }+ \e^{-\frac{\ii}{2} (k_m z+k_m^3 t+\phi+\ii\theta) }\label{eq:el_mKdV_1soliton}\\
    k_m&=\frac{(2m+1)\pi}{2K} \quad (m \in \bZ), \quad \phi,\,\theta \in \bR \label{eq:el_mKdV_wavenumber},
\end{align}
which gives 
\begin{align*}
    U(z)=\dfrac{(-1)^{m+1}k_m}{\cos \left(k_mz+k_m^3 t+\phi+\ii\theta\right)}
\end{align*}
Since $U(z)$ is holomorphic in $-K'\le \Im{z} \le K'$, $|K'k_m|<|\theta|$ is required.
Thus we find a complex solution of \eqref{eq:elliptic_mKdV} as
\begin{align}
    v(x,t)&=\frac{(-1)^{m+1}k_m}{2\ii}\left[\frac{1}{\cos(k_m x+k_m^3 t+\phi+\ii(\theta+K'))}-\frac{1}{\cos(k_m x+k_m^3 t+\phi+\ii(\theta-K'))} \right]
    \nonumber \\
    &=
    \dfrac{(-1)^{m+1}k_m \sin (k_m x+k_m^3t+\phi+\ii \theta)\sinh (k_mK')}{\cos^2(k_m x+k_m^3t+\phi+\ii\theta)+\sinh^2(k_mK')}.
\end{align}
\subsection{Double-Integro-difference equation}
The above procedures are easily extended to difference equations or double singular integral cases. 
For example, though it is fairly artificial,  we can consider a double integro-difference equation, 
\begin{align}
&\frac{1}{v_{n+2}^{t+2}(x,y)}+\frac{1}{v_{n+1}^{t+2}(x,y)}+\frac{1}{v_{n+2}^{t+1}(x,y)}-\Delta T_{y,x}\left(v_{n+1}^{t+2}(x,y)+v_{n+2}^t(x,y) +v_n^{t+1}(x,y)\right)\nonumber \\
&=\frac{1}{v_{n}^{t+1}(x,y)}+\frac{1}{v_{n+1}^{t}(x,y)}+\frac{1}{v_{n}^{t}(x,y)}-\Delta T_{y,x}\left(v_{n+2}^{t+1}(x,y)+v_{n}^{t+2}(x,y) +v_{n+1}^t(x,y)\right),
\label{eq:dint_dKdV}
\end{align}
where $\Delta$ ($0<\Delta<1$) is a constant, 
\begin{align}
 T_{y,x}(u)(x,y)&:= T_y(T_x u(x,y)), 
\end{align}
and  $T_x$ ($T_y$) is the singular integral operator with respect to $x$ ($y$) given in \eqref{def:singular_T}.
This equation is related to the discrete KdV equation \cite{Hirota-Tsujimoto}:
\begin{align}\label{eq:discrete_KdV}
    \frac{1}{v_{n+1}^{t+1}}-\frac{1}{v_n^t}+\Delta\left(v_n^{t+1}-v_{n+1}^t\right)=0.
\end{align}

To obtain solutions of \eqref{eq:dint_dKdV}, we use the following lemma.
%
%
\begin{lem}\label{lem:double_int}
  Suppose that $U(z,w)$ is holomorphic in $-\delta \le \Im{z} \le \delta$, $-\epsilon \le \Im{w}\le \epsilon$, and $\lim_{\Re{z} \to \pm\infty}U(z,w)=\mbox{const}$, and $\lim_{\Re{w} \to \pm\infty}U(z,w)=U^y_{\pm\infty}$ do not depend on $z$, then,
  \begin{align}
      &U(x+\ii\delta,y+\ii\epsilon)+U(x+\ii\delta,y-\ii\epsilon)+U(x-\ii\delta,y+\ii\epsilon)+U(x-\ii\delta,y-\ii\epsilon)-2(U^y_{\infty}+U^y_{-\infty})\nonumber \\
      &=-T_{y,x}\left[ U(x+\ii\delta,y+\ii\epsilon)-U(x+\ii\delta,y-\ii\epsilon)-U(x-\ii\delta,y+\ii\epsilon)+U(x-\ii\delta,y-\ii\epsilon) \right]\label{eq:double_int_lem}
  \end{align}
\end{lem}
 \begin{proof}
 From proposition~\ref{prop:cosh_residue}, we find that
 \begin{align*}
     U(x+\ii\delta,w)+U(x-\ii\delta,w)-U_\infty-U_{-\infty}=\frac{1}{\ii}T_x\left(U(x+\ii\delta,w)-U(x-\ii\delta,w)\right)
 \end{align*}
 If we put $G(w):=U(x+\ii\delta,w)+U(x-\ii\delta,w)$ and use the proposition again,
 noticing that $T_y(U_\infty+U_{-\infty})=0$, we obtain \eqref{eq:double_int_lem}.
 \end{proof}  
%

%
\begin{lem}\label{eq:double_lem2}
   If we put
\begin{align}
  v_n^t(x,y)&=
  \hat{v}(z_n,w_t),\quad z_n:=x+\ii\left(n+\frac{1}{2}\right)\delta,\ w_t:=y+\ii\left(t+\frac{1}{2}\right)\epsilon,   
\end{align}
and assume that $\hat{v}(z,w)$ is holomorphic in both $z$ and $w$ and $\lim_{y \to \pm \infty}\hat{v}(z_n,w_t)=1$, then, Eq.~\eqref{eq:dint_dKdV} holds if $\hat{v}(z_n,w_t)$ satisfies the discrete KdV equation, that is, 
\begin{align}\label{eq:before_dKdV}
   \frac{1}{\hat{v}(z_{n+1},w_{t+1})}- \frac{1}{\hat{v}(z_{n},w_{t})}
  +\Delta\left(\hat{v}(z_{n},w_{t+1})-\hat{v}(z_{n+1},w_{t})\right)=0.
\end{align}
\end{lem}

\begin{proof}
 Let $X(z,w):=\hat{v}(z,w+\ii\epsilon)-\hat{v}(z+\ii\delta,w)$ and $Y(z,w):=\frac{1}{\hat{v}(z,w)}-\frac{1}{\hat{v}(z+\ii\delta,w+\ii\epsilon)}$. 
Since $X(z,w)$ is holomorphic, and $\lim_{y \to \pm\infty}X(z,y+\ii t)=0$, from lemma~\ref{lem:double_int}, we have
\begin{align*}
  &X(z_{n+1},w_{t+1})+X(z_{n+1},w_{t})+ X(z_{n},w_{t+1})+ X(z_{n},w_{t})\\
  &=-T_{y,x}\left(X(z_{n+1},w_{t+1})-X(z_{n+1},w_{t})- X(z_{n},w_{t+1})+ X(z_{n},w_{t})\right)
\end{align*}
When $\hat{v}(z_n,w_n)$ satisfies \eqref{eq:before_dKdV},
$X(z,w)=\frac{1}{\Delta}Y(z,w)$, and we have Eq.~\eqref{eq:dint_dKdV}.
Thus lemma~\ref{eq:double_lem2} is proved.
\end{proof}

Once the relation between Eq.~\eqref{eq:dint_dKdV} and the discrete KdV equation is seen, the solutions of \eqref{eq:dint_dKdV} is nothing but those of the discrete KdV equation. For example its $N$ soliton solutions are given by the following proposition. 
%
%
\begin{prop}\label{prop:double_Nsoliton}
$N$ soliton solutions of Eq.~\eqref{eq:dint_dKdV} is given by the $N \times N$ determinant.    
    \begin{align}
        \hat{v}(z_n,w_t)&:=\dfrac{\tau(z_n,w_t)\tau(z_{n+1},w_{t-1})}{\tau(z_{n+1},w_t)\tau(z_{n},w_{t-1})},\label{eq:prop_double_1}\\
        \tau(z,w)&=\det_{1\le i,j\le N}\left[\delta_{ij}+c_{ij}\e^{k_i z+\omega_i w}\right],\\
        &\e^{k_i}=\frac{\Delta+(1-\Delta)p_i}{1-(1-\Delta)p_i},\quad \e^{\omega_i}=\frac{1-p_i}{p_i}, \quad c_{ij}=\frac{\gamma_j}{p_i+p_j-1},
    \end{align}
    where $\gamma_i$ ($i=1,2,...,N$) are arbitrary constants and $p_i$ ($i=1,2,...,N$) are positive constants with $p_i \ne p_j$ for $i\ne j$.
\end{prop}
\begin{proof}
The proof of proposition~\ref{prop:double_Nsoliton} is the same as that for the discrete KdV equation and is found e.g. in Ref.\cite{Hirota-Tsujimoto}.
Note that $\lim_{y\to\pm \infty}\hat{v}(z,y+\ii t)=1$.
\end{proof}

\begin{rem}
The following proposition is known to hold:
\begin{prop}\label{prop:generalisation_holomorphic}
Let \( D \) be a simply connected domain,
\( u: \partial D \to \mathbb{R} \) be a piecewise continuous function and 
\( z_0, a \in D \). 

Then, there exists a function \( K(z; z_0, a) = H(z;z_0,a) + \ii G(z;z_0)\)
given that
\(H(z;z_0,a)=\log\frac{z-z_0}{z-a} + \phi(z)\),
where \(\phi\) is holomorphic in \(D\) and satisfies \( \Im H = const \) along with \(\partial D\),
and \(G(z;z_0) = -\ii \log(z-z_0) + \psi(z)\),
where \(\psi\) is holomorphic in \(D\) and
satisfies \(\Im G = const\) along with \(\partial D\).

The \(K\) acts as an \(\) integral kernel such that
\begin{align}
U(z_0) &= \Re U(a) + \frac{1}{2\pi}\int_{\partial D} u(z) \, dK
\end{align}
becomes a holomorphic function in \( D \). Moreover, for a continuous point \( z_0' \) of \( u \),
\begin{align}
\lim_{z_0 \in D \to z_0' \in \partial D} \Im U(z_0) &= u(z_0')
\end{align}
and
\begin{align}
\lim_{z_0 \in D \to z_0' \in \partial D} \Re U(z_0) = \Re U(a) + \frac{1}{2\pi}  \int_{\partial D}u(z) \Re dH
\end{align}
hold.
\end{prop}

For example, inside the unit circle \(D = \{z|\; |z| < 1\}\): 
    \(U\) is represented as
    \begin{align}
        U(z) = \Re U(0) + \frac{\ii}{2\pi}\int_0^{2\pi}\frac{1+z\e^{-\ii \phi}}{1-z\e^{-\ii \phi}}\Im U(\e^{\ii\phi})d\phi,
    \end{align}
    and the boundary value on \( |z| = 1\) is represented as
    \begin{align}
        U(\e^{\ii \theta}) = \Re U(0) + \frac{1}{2\pi}\int_0^{2\pi}(\Im U(\e^{\ii \phi}) -  \Im U(\e^{\ii \theta}))\cot(\frac{\phi - \theta}{2})d\phi.
    \end{align}   
Therefore we can consider integro-differential equations over many kinds of complex domains.
\end{rem}

\section{Application to mathematical modeling for traffic flow}\label{sect:traffic_flow}

In this section, we consider application of integro-differential equations to a mathematical model for traffic flow.
A celebrated model for traffic flow is the Burgers equation
\begin{subequations}
\begin{align}\label{eq:Burgers}
    \rho_t(x,t)&=-\left[v_l(\rho(x,t))\rho(x,t) -D\rho_x(x,t) \right]_x,\\
   v_l(\rho(x,t))&:=V_m\left(1-\frac{\rho(x,t)}{\rho_m}\right) 
\end{align}    
\end{subequations}
where $\rho(x,t)$ is the car density at position $x$ and time $t$, $V_m$ is the maximum car velocity, $\rho_m$ is the maximum car density and $D$ is the diffusion constant which expresses fluctuation of car density. 
The term $v_l(\rho)$ denotes the average car velocity depending on the car density, we generalise it by taking into non-locality as 
\begin{align}
   \rho_t &=-2A\rho_x T(\rho_x) -V_c\rho_x +D\rho_{2x} \label{eq:traffic_coth},
\end{align}
where \(A = \frac{V_m \delta}{\rho_m} \)\cite{Satsuma-Tomoeda}\cite{H-Nukaya-Satsuma-Tomoeda},
or its further generalisation 
\begin{subequations}
\begin{align}\label{eq:traffic_ellip}
    \rho_t(x,t)&=-2A\rho_x T_k(\rho_x) -V_c\rho_x +D\rho_{2x}(x,t) \\ 
    T_k(v)&:=-\frac{1}{2\delta}\dashint_{x-\frac{2\delta K}{\pi}}^{x+\frac{2\delta K}{\pi}}\ns\left(\frac{\pi}{2\delta}(x-y), k\right)\, v(y)\,dy
\end{align}    
\end{subequations}
In fact, by assuming the fixed boundary condition for car density as 
\begin{align*}
\lim_{x \to \pm\infty}\rho(x,t)=\rho(\pm\infty),    
\end{align*}
and taking $\delta \to +0$ limit, we have
\begin{align*}
    \lim_{\delta \to +0} -2AT(\rho_x) &=\frac{V_m}{\rho_m}\dashint_{-\infty}^\infty \sgn(x-y)\rho_y(y,t)\, dt \\
    &=\frac{V_m}{\rho_m}\left\{ 2\rho(x,t)-\rho(\infty)-\rho(-\infty)\right\}.
\end{align*}
Hence if we put
\begin{align*}
    V_c=V_m\left(1-\frac{\rho(\infty)+\rho(-\infty)}{\rho_m}\right),
\end{align*}
we find
\begin{align*}
    &\lim_{\delta \to +0}-2A\rho_x T(\rho_x) -V_c\rho_x +D\rho_{2x}\\
    &=-\left[v_l(\rho(x,t))\rho(x,t) -D\rho_x(x,t) \right]_x. 
\end{align*}
Furthermore, 
\begin{align*}
    \lim_{k \to 1-0}T_k(v)=T(v).
\end{align*}
 Therefore \eqref{eq:traffic_coth} and \eqref{eq:traffic_ellip} are generalisation of the Burgers equation model \eqref{eq:Burgers}.

\subsection{Solutions for Eq.\eqref{eq:traffic_coth}}
To consider solutions for Eq.\eqref{eq:traffic_coth}, 
we set 
\begin{align}
    v(x,t):=\rho_x(x,t) \label{eq:rho_to_v}
\end{align}
Then, we have 
\begin{align}
   v_t=-\left[2Av T(v)\right]_x-V_cv_x +Dv_{2x}.
    \label{eq:coth_sing_v}
\end{align}
Let us consider a complex Burgers equation:
\begin{align}
    U_t(z)&=\left[AU^2(z)-(V_c+A(U_\infty+U_{-\infty})) U(z)+DU_z(z)\right]_z, \label{eq:comp_traffic_Burgers}\\
    U_{\pm \infty}&:=\lim_{\Re(z) \to \pm \infty} U(z)
\end{align}
From proposition~\ref{prop:cosh_residue}, a holomorphic function $U(z)$ which satisfies
$U(z^*)=U(z)^*$ and $U_\infty+U_{-\infty}$ is a real constant gives
\begin{align*}
    U(x+\ii\delta)=-T(v)(x)+\frac{U_\infty+U_{-\infty}}{2}+\ii v(x).
\end{align*}
From this relation and Eq.\eqref{eq:comp_traffic_Burgers}, we obtain Eq. \eqref{eq:coth_sing_v}.
Hence, $\Im \left[U(x+\ii\delta)\right]$ gives a solution of Eq.~\eqref{eq:coth_sing_v}.

Next, let's find specific solutions and discuss traffic flow phenomena.
Eq.\eqref{eq:comp_traffic_Burgers} is linearized through the Cole-Hopf transformation
\begin{align}
    U = \frac{D}{A}\frac{f_z}{f},
\end{align}
into the diffusion equation
\begin{align}
    f_t = Df_{zz} -(V_c +A(U_{-\infty}+U_\infty)) f_z.
\end{align}

If the dispersion relation
\begin{align}
    \omega &= Dk^2 -(V_c +A(U_{-\infty}+U_\infty))k\\
           &= Dk^2 -\left(V_m\left(1-\frac{\rho(\infty)+\rho(-\infty)}{\rho_m}\right)+A(U_{-\infty}+U_\infty)\right)k
\end{align} 
holds 
then \(f=\e^{kz + \omega t}\) is a solution for the diffusion equation.

Among linear superpositions of \(f\), we consider the form as
\begin{align}
    f(z,t) = 1 + \sum_{j=1}^N \e^{k_jz + \omega_j t + a_j},
\end{align}
where \(0<k_1<k_2<\cdots < k_N\) and \(a_1,a_2,\cdots,a_N\) are real constants.
Furthermore, to ensure holomorcity, 
we assume that for any \( j = 1, \ldots, N \), it holds that \( 0 < k_j\delta < \pi \).

Therefore 
\begin{align}
    \rho(x,t) &= \frac{D}{2A\ii}\log \frac{f(x+\ii \delta,t)}{f(x-\ii \delta,t)}\\
      &= \frac{D}{A} \arg f(x+\ii \delta,t)\\
      &=\frac{D}{A} \arctan\frac{\Im f(x+\ii \delta,t)}{ \Re f(x+\ii \delta,t)},
\end{align}
with \(\rho(-\infty,t)=0\).

\subsection{Discussion on Traffic Flow}

Let's first consider the solution for the case of \( N = 1 \).
For \( k > 0 \) and real constant \(a\), the solution is derived from \( f(x,t) = 1 + e^{kx + \omega t + a} \). 
This yields
\begin{align}
\rho(x,t) = \frac{D}{A}\arctan\frac{e^{kx+\omega t + a}\sin k\delta}{1 + e^{kx+\omega t + a}\cos k\delta}
\end{align}
The dispersion relation, given \( U_{-\infty} = \rho(-\infty) = 0 \), \( U_{\infty} = \frac{Dk}{A} \), and \( \rho(\infty) = \frac{D}{A}k\delta \), is \( \omega = -kV_m\left(1- \frac{\rho(\infty)}{\rho_m}\right) \).

\begin{figure}[bht]
\centering
\includegraphics[width=1\textwidth]{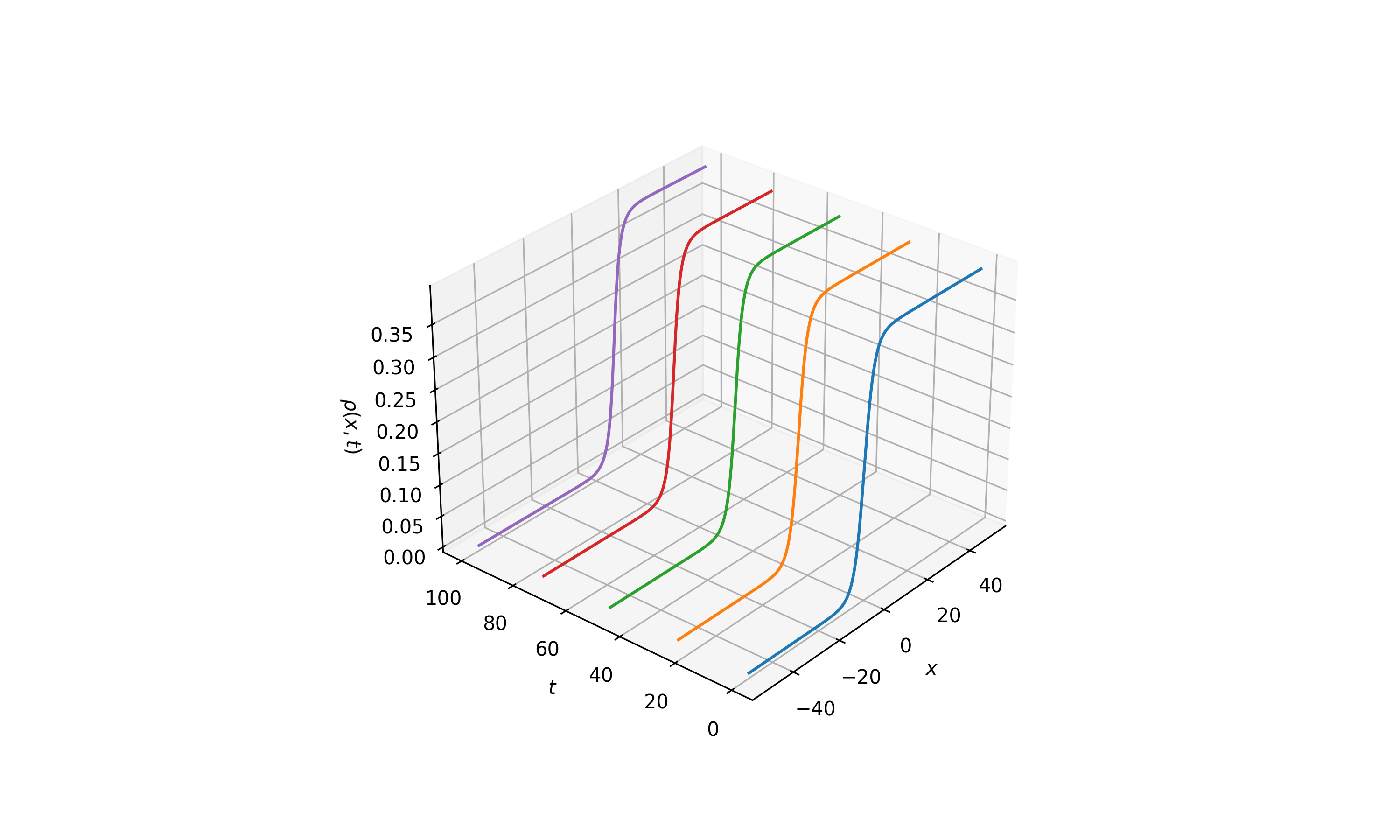}
\caption{Spacetime diagram of a 1-shock wave. 
The axes represent space \( x \), time \( t \), and density \( \rho(x, t) \), with the plot ranging from \( x = -50 \) to \( 50 \) and \( t = 0, 25,50,75,100\). The parameters are \( k=0.4 \), \( a=0 \), \( V_m=\rho_m=0.5 \), \( D=1 \), \( \delta = 5 \).
}\label{fig:1-shock-SB}
\end{figure}

This solution represents a 1-shock wave propagating at a steady velocity, which can be interpreted as modeling the movement of vehicles in their direction of travel.
Fig. \ref{fig:1-shock-SB} shows an example of this solution.

Now, consider the solution \(\rho(x,t)\) which is derived from \(f(z,t) = 1 + \sum_{j=1}^N \e^{k_jz + \omega_j t + a_j}\) in the general case of arbitrary \(N\). 
This solution represents the transition of \(N\) shock waves into a single shock wave, which is demonstrated as follows.
By transforming to a coordinate system moving at a speed \(-V_c\), it is deduced from the conditions \( \omega_j < 0 \) for \( j=1,2,...,N-1 \) and \( \omega_N = 0 \), leading to the limit 
\begin{align}
    \lim_{t \rightarrow \infty} \rho_x(x,t) = \frac{D}{A}\frac{k_N \sin{k_N\delta}}{\cosh{(k_N x + a_N)} + \cos{k_N\delta}}.
\end{align}
Thus it is concluded that the solution converges to a single shock wave moving in the direction of travel at the speed of \(V_c\). 

The interaction between free-flow and congested traffic states leads to the absorption of free-flow into congested traffic, resulting in the expansion of congestion. 
Confirm this phenomenon using density and flux.
Regarding equation \eqref{eq:traffic_coth} as the equation of continuity \(\rho_t + J_x = 0\) provides the flux
\begin{align}
    J(x,t) = -\frac{D}{A} \int_{-\infty}^x \frac{\Re f(x+\ii \delta,t) \pdv{\Im f(x+\ii \delta,t)}{t}-\Im f(x+\ii \delta,t) \pdv{\Re f(x+\ii \delta,t)}{t}}{(\Re f(x+\ii \delta,t))^2 + (\Im f(x+\ii \delta,t))^2}dx.
\end{align}

Figure \ref{fig:4-shock-SB} represents the behavior of \(N\)-shock waves, illustrating that low-density waves, considered as free-flow state, are absorbed into a congested state after colliding with high-density waves, perceived as congested state. 
\begin{figure}[bht]
\centering
\includegraphics[width=1\textwidth]{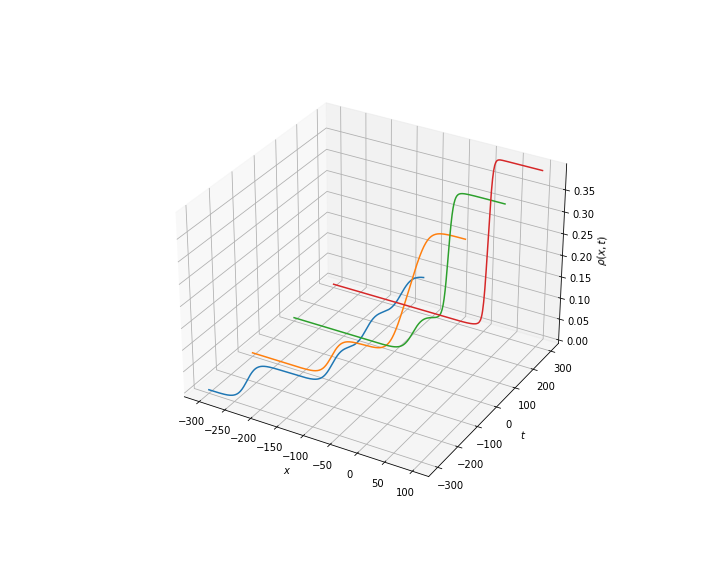}
\caption{Spacetime diagram of a \(N\)-shock wave (\(N=4\)). 
The axes represent space \( x \), time \( t \), and density \( \rho(x, t) \), with the plot ranging from \( x = -300 \) to \( 100 \) and \( t = -300, -100,100,300\).
The parameters are \( k_1=0.1,k_2=0.2,k_3=0.3,k_4=0.4 \), \( a_1=a_2=a_3=10,a_4=20 \), \( V_m=\rho_m=0.5 \), \( D=1 \), \( \delta = 0.01 \).
}\label{fig:4-shock-SB}
\end{figure}

The flux corresponding to Fig. \ref{fig:4-shock-SB} is shown in Fig.\ref{fig:4-shock-JSB}. 
At \(t = -300, -100\), it is observed that in low-density regions, the flux increases with density, indicating the presence of a free-flow state. 
However, at \(t = 100, 300\), a decrease in flux is observed, indicating that the free-flow state is changing into a congested state.

Furthermore, under the condition \( \rho_m = \rho_\infty \), since \(V_c=0\), the shock wave does not travel after its formation. This can be considered as a description of the deadlock phenomenon, where the road capacity ahead is full, and all cars stop.

\begin{figure}[bht]
\centering
\includegraphics[width=1\textwidth]{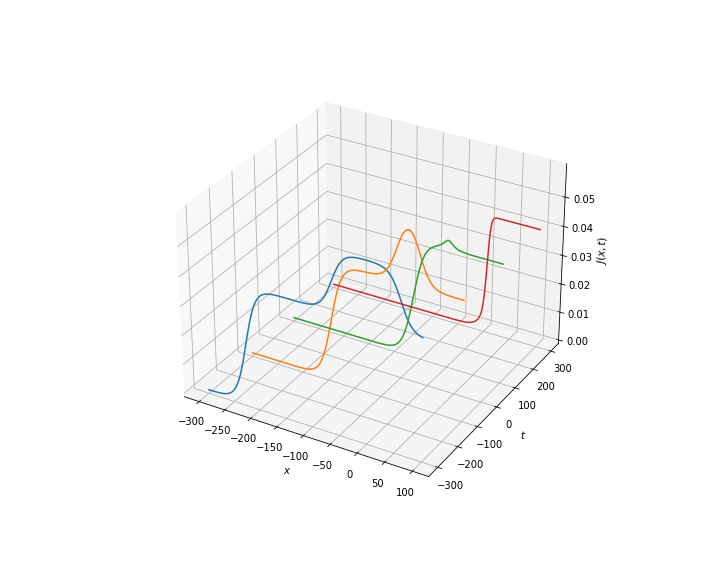}
\caption{Spacetime diagram of the flux corresponding to \(N\)-shock wave (\(N=4\)). 
The axes represent space \( x \), time \( t \), and density \( J(x, t) \), with the plot ranging from \( x = -300 \) to \( 100 \) and \( t = -300, -100,100,300\).
The parameters are \( k_1=0.1,k_2=0.2,k_3=0.3,k_4=0.4 \), \( a_1=a_2=a_3=10,a_4=20 \), \( V_m=\rho_m=0.5 \), \( D=1 \), \( \delta = 0.01 \).
These are the same settings as in Fig. \ref{fig:4-shock-JSB}.
}\label{fig:4-shock-JSB}
\end{figure}

Next, investigate the dependence on the parameter \(\delta\).
As pointed out by Satsuma and Tomoeda in \cite{Satsuma-Tomoeda}, the integral kernel corresponding to the Burgers equation is the sign function, suggesting that the driver's perception distance is considered to be equivalent to infinite. 
In contrast, the generalised equation \eqref{eq:traffic_coth} with the \(\coth\) function as the integral kernel implies that the driver's perception distance is effectively finite, and a larger \(\delta\) leads to a shorter perception distance. 

Figure~\ref{fig:4shock-dep-delta} plots the dependence of density on \(\delta\).
During the collision, the gradient of the shock wave front increases for a larger \(\delta\).
It is seen from the fact that the maximum value of \(\rho_x \) of a single shock wave is proportional to \(\tan\frac{k_N\delta}{2}\).

This means that due to the narrow perception distance, vehicles merge into the congestion without slowing down, leading to a rapid deceleration. 
Therefore, we can understand the reason why the gradient of the wave front increases.

\begin{figure}[bht]
\centering
\begin{minipage}{0.5\hsize}
\begin{center}
\includegraphics[width=1.\textwidth]{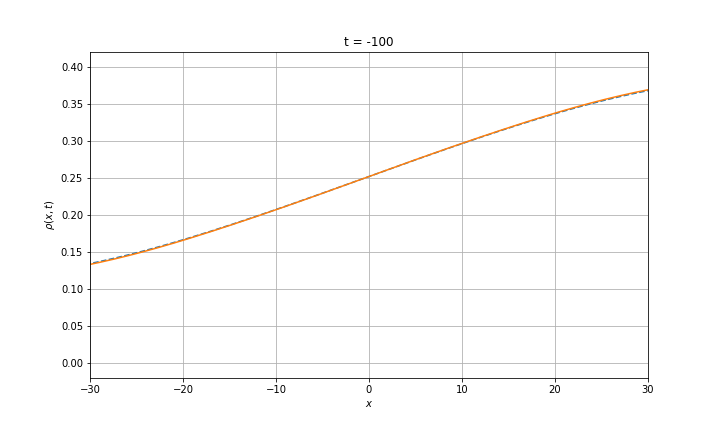}
\hspace{1.0cm}(a) $t = -100$
\end{center}
\end{minipage}%
\begin{minipage}{0.5\hsize}
\begin{center}
\includegraphics[width=1.\textwidth]{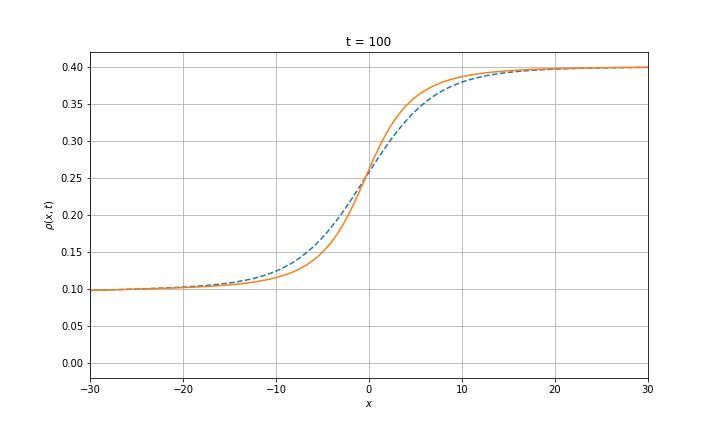}
\hspace{1.0cm}(b) $t = 100$)
\end{center}%
\end{minipage}
\begin{minipage}{0.5\hsize}
\begin{center}
\includegraphics[width=1.\textwidth]{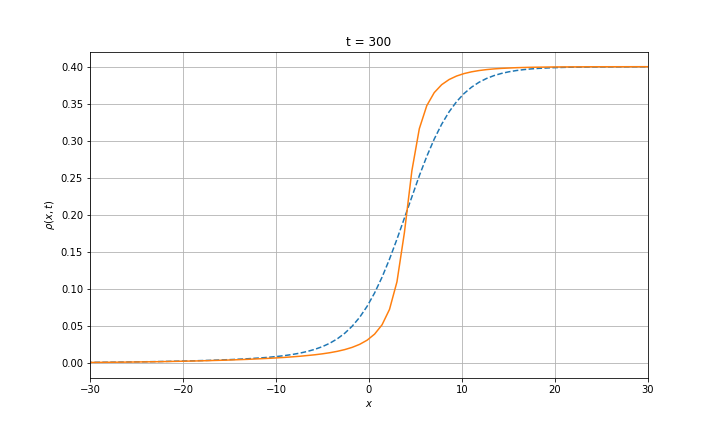}
\hspace{1.0cm}(c) $t = 300$
\end{center}%
\end{minipage}
\caption{These figures are snapshots of the density \(\rho(x,t)\) with a fixed \(t\). 
The horizontal axis represents \(x\), and the vertical axis is \(\rho(x,t)\). 
The parameters, other than \(\delta\), are the same as in Fig \ref{fig:4-shock-SB}, but \(x\) is limited to the range from \(-30\) to \(30\). 
The dashed line represents \(\delta = 0.01\), and the solid line represents \(\delta = 7\). 
Regarding time, (a) is at \(t = -100\), (b) is at \(t = 100\), and (c) is at \(t = 300\).
}
\label{fig:4shock-dep-delta}.
\end{figure}

\subsection{Traffic flow model with elliptic singularity}
%
%
Finally we discuss a generalised traffic flow model with singular integral of a Jacobi elliptic function.
As in the previous subsection, we consider a  dimensionless equation
\begin{subequations}
    \begin{align}
   v_t&=-\left[2v T_k(v)\right]_x-V_cv_x +Dv_{2x}
    \label{eq:ellip_sing_v} \\
   T_k(v)&:=-\frac{1}{\pi}\dashint_{x-K}^{x+K} \ns(x-y,k)v(y)\,dy
\end{align}
\end{subequations}
The maximum velocity $V_m$, in general, could change according to the circumstance, so we allow it to depend on time $t$ and position $x$, so that $V_c=V_c(x,t)$ 

Since 
\begin{align*}
    \lim_{k \to 1-0}\ns(x,k)=\coth(x),
\end{align*}
 \eqref{eq:ellip_sing_v} is a generalisation of \eqref{eq:traffic_coth},
and here we treat another limit $k \to +0$.
In this limit, it holds that
\begin{align*}
    \lim_{k \to +0}\ns(x,k)= \frac{1}{\sin x},\quad
     \lim_{k \to +0}K=\frac{\pi}{2},\quad \lim_{k \to 0}K'=+\infty.
\end{align*}
%
%
\begin{prop}\label{lem:KtoInfty}
If a function $f(x+\ii \eta)$ ($x,\,\eta \in \bR$) does not depend on modulus \(k\) and converges in the limit $\eta \to +\infty$ as
\begin{align*}
    \lim_{\eta \to +\infty} f(x+\ii \eta)=f_\infty(x),
\end{align*}
then,
\begin{align}
&\lim_{k \to +0}\frac{1}{\pi}\int_0^{K'}\frac{k}{\dn(\eta,k')}f(x+\ii \eta)\, d\eta =\frac{2}{\pi}f_\infty(x)\label{eqe:Iv_infty}
\end{align}
\end{prop}

\begin{proof}
\begin{align*}
C&:=\lim_{k \to +0}\frac{1}{\pi}\int_0^{K'}\frac{k}{\dn(\eta,k')}f(x+\ii \eta)\, d\eta\\
&=\lim_{k \to +0}\frac{1}{\pi}\int_0^{1}\frac{k K'}{\dn(K'y,k')}f(x+\ii K'y)\, dy. 
\end{align*}
Using asymptotic expansion, we notice
\begin{align*}
    \frac{1}{\dn(K'y,k')}&=(\cosh{K'y})\left(1+o(k)\right),\\
    K'&=\log\frac{4}{k} + o(k).
\end{align*}
Hence 
\begin{align*}
 C&=\lim_{k \to +0}\frac{1}{\pi}\int_0^1 -k \log(\frac{k}{4}) \frac{1}{2}\e^{-(\log \frac{k}{4})y} f(x-\ii(\log \frac{k}{4})y)\, dy\\
 &=\lim_{k \to +0}\frac{1}{\pi}\int_0^1 -k \log(\frac{k}{4}) \frac{1}{2}\e^{-(\log \frac{k}{4})y}f_{\infty}(x) \, dy \\
 &\quad +\lim_{k \to +0}\frac{1}{\pi}\int_0^1 -k \log(\frac{k}{4}) \frac{1}{2}\e^{-(\log \frac{k}{4})y}\left\{f(x-\ii(\log \frac{k}{4})y)-f_{\infty}(x)\right\}\, dy.
\end{align*}
The first term converges to $\frac{2}{\pi}f_\infty(x)$ and the second term converges to $0$.
Thus, we obtain
\begin{align*}
 C=\frac{2}{\pi}f_\infty(x)   
\end{align*}
\end{proof}

Now we consider an integro-differential equation:
\begin{align}\label{eq:last_sin_eq}
v_t=\left[-2T_0(v(x,t))v(x,t)-V_c(x,t)v(x,t)+Dv_x(x,t) \right]_x   
\end{align}
where 
\begin{align*}
    T_0(v):=-\frac{1}{\pi}\int_{x-\frac{\pi}{2}}^{x+\frac{\pi}{2}}\frac{1}{\sin(x-y)}v(y,t)\, dy
\end{align*}
As an example in which an exact solution of this limit is obtained,  we consider the case where $V_c(x,t)$ ($t>0$) has periodic fluctuation in $x$ and converges to a constant $V_c(\infty)$. 
The following function satisfies this condition and can be treated analytically.
\begin{align}
   &V_c(x,t)=V_c(\infty)+\frac{D}{4\pi}\Re\left[\e^{\ii(x-V_c(\infty) t +\theta)-D t}\log \left(1+\e^{-2\ii(x-V_c(\infty) t +\theta)-2D t}\right) \right]
   \label{eq:sine_real_Vc}
\end{align}
%
%
 
If $U(z,t)$ is a holomorphic function of $z$ and satisfies $U(z^*,t)=U(z,t)^*$, then, putting $U(x+\ii K',t)=u(x,t)+\ii v(x,t)$, we have from corollary~\ref{cor:elliptic_relation}, 
\begin{align}
    &u(x,t)=-T_k(v(x,t))\nonumber \\
    &\quad +\frac{1}{\pi}\int_0^{K'}
\frac{k}{\dn(\eta,k')}\Re\left[U(x+K+\ii \eta)+U(x-K+\ii \eta)\right]\, d\eta \label{eq:Te_real_relation}
\end{align}
From \eqref{eq:Te_real_relation}, we notice that \eqref{eq:ellip_sing_v} is related to the following equation.
\begin{align}
    U(z,t)_t=\left[U(z,t)^2-\tilde{V_c}U(z,t)+DU_z(z,t)  \right]_z \label{eq:complex_Te}
\end{align}
Here, $x=\Re[z]$ and 
\begin{subequations}
\begin{align}
    \tilde{V}_c(z,t)&:=V_c(z,t)+2C(z,t) \label{elip_tilde} \\
    C(x,t)&:=\frac{1}{\pi}\int_0^{K'}
\frac{k}{\dn(\eta,k')}\Re\left[U(x+K+\ii \eta,t)+U(x-K+\ii \eta,t)\right]\, d\eta.
\end{align}
\end{subequations}
Our strategy is to determine $U(z,t)$ so that it gives $\tilde{V}_c(x,t)=V_c(\infty)$, that is, 
\begin{align}\label{eq:real_condition}
 &2C(x,t)=V_c(\infty)-V_c(x,t).   
\end{align}
By changing variable
\begin{align}
    U(z,t)=D\left(\log F(z,t)\right)_z,\label{eq:ellip_Fzt}
\end{align}
\eqref{eq:complex_Te} turns to 
\begin{align}
    F_t=DF_{2z}-V_c(\infty)F_z. \label{eq:elip_F_linear}
\end{align}
A general solution of \eqref{eq:elip_F_linear} is
\begin{align}
    F(z,t)=\int \mu(dk)\, \e^{\omega(k) t+kz}
\end{align}
where $\mu(dk)$ denotes arbitrary measure and $(\omega,k)$ satisfy the dispersion relation
\begin{align}
    \omega(k)=Dk^2-V_c(\infty)k. \label{eq:sine_disp}
\end{align}
As a simple case, we consider
\begin{align}
    F(z,t)=1+A(K')\e^{ \omega t}\cos(qz-\nu t + \theta) 
\end{align}
where $A(K')$ is a real constant depending on $K'$, and $q$ ($q>0$), $\omega$, $\nu$, $\theta$ are real constants which satisfy
\begin{align}
    \omega=-q^2D,\quad \nu=qV(\infty).
\end{align}
Hence
\begin{align*}
    U(z,t)=\dfrac{-D A(K')\e^{ \omega t}\sin(qz-\nu t + \theta) }{1+A(K')\e^{ \omega t}\cos(qz-\nu t + \theta) }.
\end{align*}
When we put $A(K')=2\e^{-K'}$ and $q=1$ for simplicity, then similarly to proposition~\ref{lem:KtoInfty},
\begin{align*}
C(x,t)&:=\lim_{k \to +0} \frac{1}{\pi}\int_0^{K'}
\frac{k}{\dn(\eta,k')}\Re\left[U(x+K+\ii \eta,t)+U(x-K+\ii \eta,t)\right]\, d\eta\\
&=\lim_{k \to +0}\frac{1}{\pi}\int_0^{K'}4\e^{-K'}\cosh \eta
\Re\left[U(x+\frac{\pi}{2}+\ii \eta,t)+U(x-\frac{\pi}{2}+\ii \eta,t)\right]\, d\eta.
\end{align*}
Here
\begin{align*}
U(x\pm\frac{\pi}{2}+\ii \eta,t)&=\dfrac{-2D\e^{-K'}\e^{\omega t}(\sin \zeta_\pm \cosh \eta+\ii \cos \zeta_\pm \sinh \eta)}{1+2\e^{-K'}\e^{\omega t}(\cos \zeta_\pm\cosh \eta-\ii\sin \zeta_\pm\sinh \eta)}   \\
&\sim \dfrac{-\ii D\e^{\omega t-K'+\eta-\ii \zeta_\pm}}{1+\e^{\omega t-K'+\eta-\ii\zeta_\pm}}\quad (k \to +0),
\end{align*}
where $\zeta_\pm:=\zeta\pm\frac{\pi}{2}$, and $\zeta:= x-\nu t +\theta$.
Then,
\begin{align*}
C(x,t)&=\lim_{K' \to \infty}\frac{2}{\pi}\int_0^{K'}\e^{-K'+\eta}\Re\left[ \dfrac{-\ii D\e^{\omega t-K'+\eta-\ii \zeta_+}}{1+\e^{\omega t-K'+\eta-\ii\zeta_+}}+\dfrac{-\ii D\e^{\omega t-K'+\eta-\ii \zeta_-}}{1+\e^{\omega t-K'+\eta-\ii\zeta_-}}\right]\,d\eta\\
&=\frac{2}{\pi}\int_0^1\Re\left[(-\ii D)\{\frac{(\e^{-\ii\zeta_+})s}{\e^{-\omega t}+(\e^{-\ii\zeta_+})s}  +\frac{(\e^{-\ii\zeta_-})s}{\e^{-\omega t}+(\e^{-\ii\zeta_-})s} \}\right]\,ds\\
&=\frac{2}{\pi}\int_0^1\Re\left[(-\ii D)\{ \frac{-\ii(\e^{-\ii\zeta})s}{\e^{-\omega t}-\ii(\e^{-\ii\zeta})s}  +\frac{\ii (\e^{-\ii\zeta})s}{\e^{-\omega t}+\ii (\e^{-\ii\zeta})s} \}\right]\,ds\\
&=\frac{2D}{\pi}\int_0^1\Re\left[-\frac{(\e^{-\ii\zeta})s}{\e^{-\omega t}-\ii(\e^{-\ii\zeta})s} +\frac{(\e^{-\ii\zeta})s}{\e^{-\omega t}+\ii (\e^{-\ii\zeta})s} \right]\,ds\\
&=\frac{2D}{\pi}\Re\left[\e^{\ii \zeta+\omega t}\log \left(1+\e^{-2\ii \zeta+2\omega t}\right) \right]
\end{align*}

From corollary~\ref{cor:elliptic_relation}, we find that
\begin{align*}
   &\lim_{K' \to \infty}\Re[U(x+\ii K',t)]
    =-\lim_{K' \to \infty}T_0(\Im[U(x+\ii K',t)])+C(x,t) 
\end{align*}
Since $U(x,t)$ satisfies \eqref{eq:complex_Te} for $\tilde{V}_c=V_c(\infty)$, its imaginary part $v(x,t):=\lim_{K' \to \infty}T_0(\Im[U(x+\ii K',t)])$ satisfies
\begin{align*}
    v_t(x,t)&=\left[2\lim_{K' \to \infty}\Re[U(x+\ii K',t)]-V_c(\infty)v(x,t)+Dv_x(x,t) \right]_x\\
    &=\left[-2T_0(v(x,t))v(x,t)+2C(x,t)v(x,t)-V_c(\infty)v(x,t)+Dv_x(x,t) \right]_x\\
    &=\left[-2T_0(v(x,t))v(x,t)-V_c(x,t)v(x,t)+Dv_x(x,t) \right]_x
\end{align*}
Similar estimation for $\lim_{K' \to \infty}\Im[U(x+\ii K',t)]$ gives
\begin{align}\label{eq:ellip_sing_solex}
    v(x,t)&=\Im\left[ \frac{-\ii D\e^{-\ii \zeta}}{\e^{-\omega t}+\e^{-\ii\zeta}}\right]\nonumber\\
    &=-\frac{D}{2}\frac{\e^{\omega t}+\cos(x-\nu t+\theta)}{\cosh \omega t + \cos (x-\nu t+\theta)}.
\end{align}
Therefore, we find that \eqref{eq:ellip_sing_solex} is an exact solution of \eqref{eq:last_sin_eq}.

\section{Concluding Remarks}
In this article we consider a traffic model which is given by an integro-differential equation  \eqref{eq:traffic_coth}.
The model is a non-local extension of the celebrated Burgers equation.
We showed that generally, nonlinear partial differential equations with singular integrals can be constructed by considering the boundary values of holomorphic function in a complex domain, based on the residue theorem. 
Several examples of integrable integro-differential equations including double singular integrals and elliptic singularity kernel have been provided.
Then, we discussed the traffic model with a singular integral and provided \(N\)-shock wave solutions and fluxes for the extended version of the Burgers equation with singular integrals and discussed the phenomenon of deadlock and the impacts of long-distance correlations.

Although we mainly investigated holomorphic solutions of the equation, a solution with singularity of poles would present some interesting features such as a blowing-up phenomenon. Investigation of singularity as well as the detailed examination of the solutions of the traffic flow model with elliptic kernel are the problems we would like to address in future.

\section*{acknowledgement}
We would like to thank Professors Chihiro Matsui, Junkichi Satsuma, Tetsuji Tokihiro, and  Ralph Willox for their encouragement and useful discussions and comments.
This work is supported in part by JST SPRING, Grant Number JPMJSP2108.

\clearpage
\bibliography{ref}   

\bibliographystyle{unsrt}  






\end{document}